\documentclass[twoside,leqno]{article}  
\usepackage{ltexpprt}
\usepackage[T1]{fontenc}
\usepackage[latin9]{inputenc}
\usepackage{todonotes}
\usepackage{amsmath}
\usepackage{hyperref}
\usepackage{wrapfig}
 \usepackage{subcaption}
 \usepackage{xspace}
 \usepackage{ushort}
 \usepackage{setspace}
 \usepackage{float}
 \usepackage[algo2e,ruled,linesnumbered,noline]{algorithm2e} 
 \usepackage{placeins}
 \usepackage{enumerate}
 \usepackage{graphicx}
 \graphicspath{{/Users/admin/workspace/dynBC_journal/graphics/}}
 
\def\squareforqed{\hbox{\rlap{$\sqcap$}$\sqcup$}}
\def\qed{\ifmmode\squareforqed\else{\unskip\nobreak\hfil
\penalty50\hskip1em\null\nobreak\hfil\squareforqed
\parfillskip=0pt\finalhyphendemerits=0\endgraf}\fi}

\newcommand{\ie}{i.\,e.\xspace}
\newcommand{\eg}{e.\,g.\xspace}
\newcommand{\etal}{et al.\xspace}

\newcommand{\vd}{$\mathit{VD}$\xspace}

\newcommand{\rk}{\textsf{RK}\xspace}
\newcommand{\ia}{\textsf{IA}\xspace}
\newcommand{\iaw}{\textsf{IAW}\xspace}
\newcommand{\da}{\textsf{DA}\xspace}
\newcommand{\daw}{\textsf{DAW}\xspace}
\newcommand{\dad}{\textsf{DAD}\xspace}
\newcommand{\dadw}{\textsf{DADW}\xspace}
\newcommand{\vda}{$\tilde{\mathit{VD}}$\xspace}
\newcommand{\upvd}{\textsf{updateApprVD-W}\xspace}
\newcommand{\sssp}{\textsf{updateSSSP-W}\xspace}
\newcommand{\upvdu}{\textsf{updateApprVD-U}\xspace}
\newcommand{\ssspu}{\textsf{updateSSSP-U}\xspace}

\begin{document}
\setstretch{1.175}

\title{\Large Approximating Betweenness Centrality in Fully-dynamic Networks\thanks{Parts of this paper have been published in a preliminary form in the Proceedings of the Seventeenth Workshop on Algorithm Engineering and Experiments
(ALENEX 2015)~\cite{DBLP:conf/alenex/BergaminiMS15} and the Proceedings of the 23rd Annual European Symposium
on Algorithms (ESA 2015)~\cite{DBLP:conf/esa/BergaminiM15}.}
}
\author{Elisabetta Bergamini\thanks{Institute of Theoretical Informatics, Karlsruhe Institute of Technology (KIT), Germany.}
\and 
Henning Meyerhenke\footnotemark[2]}
 
\date{}

\maketitle

%\pagenumbering{arabic}
%\setcounter{page}{1}%Leave this line commented out.

\begin{abstract}
Betweenness is a well-known centrality measure that ranks the nodes of a network according to their participation in shortest paths.
Since an exact computation is prohibitive in large networks, several approximation algorithms have been proposed. Besides that, recent years
have seen the publication of dynamic algorithms for efficient recomputation of betweenness in networks that change over time.

In this paper we propose the first betweenness centrality approximation algorithms with a provable guarantee on the maximum approximation error for dynamic networks.
Several new intermediate algorithmic results contribute to the respective approximation algorithms:
(i) new upper bounds on the vertex diameter, (ii) the first fully-dynamic algorithm for updating an approximation of the vertex
diameter in undirected graphs, and (iii) an algorithm with lower time complexity for updating single-source shortest paths in unweighted
graphs after a batch of edge actions.

Using approximation, our algorithms are the first to make in-memory computation of betweenness in dynamic networks with millions of edges feasible. Our experiments show that our algorithms can achieve substantial speedups compared to recomputation, up to 
several orders of magnitude.
Moreover, the approximation accuracy is usually significantly better than the theoretical guarantee in terms of absolute error. More importantly, for reasonably small approximation error thresholds, the rank of nodes is well preserved, in particular for nodes with high betweenness.\\[0.25ex]

\end{abstract}

% page limit for ALENEX: 10 pages excluding title page and bibliography

\section{Introduction}
\label{sec:introduction}
% Context
The algorithmic analysis of complex networks has become a highly active research area. 
One important task in network analysis is to rank
nodes by their structural importance using centrality measures.
%
% Definition and applications
\emph{Betweenness centrality} (BC) is among the widely used measures; it ranks the importance of nodes based on 
their  participation in the shortest paths of the network. 
More formally, let the graph $G$ represent a network with $n$ nodes and $m$ edges.
Naming $\sigma_{st}$ the number of shortest paths from a node $s$ to a node $t$ and $\sigma_{st}(v)$ the number of shortest paths from $s$ to $t$ that go through $v$, the (normalized) BC
of $v$ is defined as~\cite{citeulike:1025135}:
$c_{B}(v)=\frac{1}{n(n-1)}\sum_{s\neq v \neq t}\frac{\sigma_{st}(v)}{\sigma_{st}}$.
Nodes with high betweenness can be important in routing, spreading processes
and mediation of interactions.
Depending on the context, this can mean, for example, finding the
most influential persons in a social network, the key infrastructure
nodes in the internet, or super spreaders of a disease.

Since BC depends on \textit{all} shortest paths, its exact computation is expensive: the best known 
algorithm~\cite{Brandes01betweennessCentrality} is quadratic in the number of nodes for sparse networks and cubic for dense networks, which is prohibitive for networks with hundreds of thousands of nodes. Many graphs of interest, however, such as web 
graphs or social networks, have millions or even billions of nodes and edges. 
Thus, recent years have seen the publication of several approximation algorithms
that aim to reduce the computational effort, while finding BC values
that are as close as possible to the exact
ones. Good results have been obtained in this regard; in particular,
a recent algorithm by Riondato and Kornaropoulos (\textsf{RK})~\cite{DBLP:conf/wsdm/RiondatoK14}  gives probabilistic guarantees on the quality of the approximation. It is described in more detail in Section~\ref{sec:rk} since we build our algorithms on this method.

\paragraph{Motivation.}
Large graphs of interest, such as the Web and social networks,
change continuously. Thus, in addition to the identification of important 
nodes in a static network, an issue of great interest is the dynamic behavior of centrality 
values in networks that change over time. So far, there have been no approximation
algorithms that efficiently update BC scores rather than recomputing them from scratch.
Several methods have been proposed to update the BC values after a graph modification, which for some of the
algorithms can only be one edge insertion and for others can also be
one edge deletion. However, all of these approaches are exact and have
a worst-case time complexity which is the same as Brandes's algorithm (\textsf{BA})
\cite{Brandes01betweennessCentrality} on general graphs and a memory footprint of $\Omega(n^2)$. 

\paragraph{Contribution.}
In this paper, as our main contribution, we present the first approximation algorithms for BC in graphs that change over time.
Such graphs may be directed or undirected, weighted or unweighted.
We consider two dynamic scenarios, an incremental one (\ie only edge insertions or weight decreases are allowed)
and a fully-dynamic one, which also handles edge deletions or weight increase operations.
After each batch of edge actions, we assert the same guarantee as the static \textsf{RK} algorithm: the 
approximated BC values differ by at most $\epsilon$
from the exact values with probability at least $1-\delta$, where
$\epsilon$ and $\delta$ can be arbitrarily small constants. 
Running time and memory required depend on how tightly the error should be bounded.

Besides resampling as few shortest paths as possible, several new intermediate algorithmic results 
contribute to the speed of the respective new approximation algorithms:
(i) In Section~\ref{sec:new_vd_approx} we propose new upper bounds on the vertex diameter \vd (\ie number of nodes in the shortest path(s) with the maximum number of nodes). These bounds vary depending on the graph type (weighted vs unweighted, directed vs undirected). 
Their usefulness stems from the fact that the new bounds can often improve the one used in the \rk algorithm~\cite{DBLP:conf/wsdm/RiondatoK14} and thereby significantly reduce the number of samples necessary for the error guarantee.
(ii) In Section~\ref{sec:dyn-algo}, besides detailing the BC approximation algorithms, we also present the first fully-dynamic algorithms for updating an approximation of \vd in undirected graphs.
(iii) As part of the BC approximation algorithms, we propose an algorithm with lower time complexity 
for updating single-source shortest paths in unweighted graphs after a batch of edge actions.

Our experimental study shows that our algorithms are the first to make in-memory computation of a betweenness ranking practical for large dynamic networks.
With approximation we achieve a much improved scaling behavior compared to exact approaches (also dynamic ones),
enabling us to update approximate betweenness scores in a network with 36 million edges in a few seconds on typical workstation hardware.
Moreover, processing batches of edge actions, our algorithms yield significant speedups (several orders of magnitude) compared to restarting the approximation with \textsf{RK}.
Regarding accuracy, our experiments show that the estimated absolute errors are always lower than the guaranteed ones. 
For nodes with high betweenness, also the rank of nodes is well preserved, even when relatively few shortest paths are sampled.

\section{Related work}
\label{sec:related_work}
\paragraph*{Static BC algorithms - exact and approximation.}

The fastest existing method for the exact BC computation, \textsf{BA},
requires $\Theta(nm)$ operations for unweighted graphs
and $\Theta(nm+n^{2}\log n)$ for graphs with positive edge weights~\cite{Brandes01betweennessCentrality}.
\textsf{BA} computes for every node $s\in V$
a slightly modified version of a single-source shortest-path tree
(SSSP tree), producing for each $s$ the directed acyclic
graph (DAG) of \emph{all} shortest paths starting at $s$. Using
the information contained in the DAGs, \textsf{BA} computes the \textit{dependency}
 $\delta_{s}(v)$ for each node $v$, that is
the sum over all nodes $t$ of the fraction of shortest paths between
$s$ and $t$ that $v$ is internal to. The betweenness
of each node $v$ is simply the \emph{sum} over all sources $s\in V$
of the dependencies $\delta_{s}(v)$. Therefore, we can
see the dependency $\delta_{s}(v)$ as a \emph{contribution}
that $s$ gives to the computation of $c_{B}(v)$.
Based on this concept, some algorithms for an
\emph{approximation} of BC have been developed.
Brandes and Pich~\cite{DBLP:journals/ijbc/BrandesP07} propose to
approximate $c_{B}(v)$ by extrapolating it from the contributions
of a \emph{subset} of source nodes, also
called \emph{pivots}. Selecting the pivots uniformly at random, the
approximation can be proven to be an unbiased
estimator for $c_{B}(v)$ (i.e.\, its expectation is equal to $c_{B}(v)$).
In a subsequent work, Geisberger \etal~\cite{DBLP:conf/alenex/GeisbergerSS08}
notice that this can overestimate BC scores
of nodes close to the pivots. To limit
this bias, they introduce
a \emph{scaling function} which gives less importance to
contributions from pivots that are close to the node. 
%In particular, the authors propose two possibilities: a linear scaling function,
%where the effect of the contribution of a pivot on a node $v$ increases
%linearly with the distance between the pivot and $v$, and a bisection
%scaling function, which considers only contributions of pivots that
%are ``far enough'' from $v$ and ignores contributions from other
%pivots. 
%In contrast to the previous works, which consider the problem of
%estimating the betweenness of \textit{all} nodes in the graph, 
Bader \etal~\cite{DBLP:conf/waw/BaderKMM07} approximate
the BC of a specific node only, based on 
an adaptive sampling technique that reduces the number of pivots for nodes with high centrality.
Chehreghani~\cite{DBLP:journals/cj/Chehreghani14} proposes alternative sampling techniques that
try to minimize the number of samples.
Different from the previous approaches is the approximation algorithm
by Riondato and Kornaropoulos~\cite{DBLP:conf/wsdm/RiondatoK14},
which samples a \emph{single} random shortest path at each iteration.
This approach allows a theoretical guarantee on the quality of approximation (see Section~\ref{sec:rk}).
Because of this guarantee, we use it as a 
building block of our new approach and refer to it as \textsf{RK}.

\paragraph*{Exact dynamic algorithms.}
Dynamic algorithms update the betweenness values of all nodes
in response to a modification on the graph, which might be an edge/node
insertion, an edge/node deletion or a change in an edge's weight. 
The first published approach
is \textsf{QUBE} by Lee \etal~\cite{DBLP:conf/www/LeeLPCC12}, which relies on the decomposition
of the graph into biconnected components. The approach has been extended to node updates in~\cite{DBLP:journals/im/GoelSIG15}.
The approach proposed
by Green \etal~\cite{DBLP:conf/socialcom/GreenMB12} for unweighted graphs 
maintains a structure with the previously calculated BC values and
additional information, like the distance
of each node $v$ from every source $s\in V$ and the list of \textit{predecessors}, i.e.\ the nodes
immediately preceding $v$ in all shortest paths from $s$ to $v$. Using this information,
it tries to limit the recomputations to the nodes whose
betweenness has actually been affected. Kourtellis
\etal~\cite{DBLP:journals/tkde/KourtellisMB15} modify the
approach by Green \etal~\cite{DBLP:conf/socialcom/GreenMB12} in
order to reduce the memory requirements. Instead of
storing the predecessors of each node $v$ from each possible
source, they recompute them every time the information is required. 

Kas \etal~\cite{DBLP:conf/asunam/KasWCC13} extend an existing algorithm for
the dynamic APSP problem by Ramalingam and Reps~\cite{Ramalingam92anincremental}
to also update BC scores.
The recent work by Nasre 
\etal~\cite{DBLP:conf/mfcs/NasrePR14} contains the first dynamic
algorithm for BC (\textsf{NPR}) which is asymptotically faster than recomputing 
from scratch on certain inputs. In particular, when only edge insertions are allowed and
the considered graph is sparse and weighted, their algorithm
takes $O(n^2)$ operations, whereas \textsf{BA} requires
$O(n^2\log n)$ on sparse weighted graphs. Pontecorvi and Ramachandran~\cite{DBLP:journals/corr/PontecorviR15} extend existing fully-dynamic APSP algorithms with new data structures
to update \textit{all} shortest paths (APASP) and then recompute dependencies as in \textsf{BA}.

All dynamic algorithms mentioned perform better than
recomputation on certain inputs. Yet, none of them
has a worst-case complexity better than \textsf{BA} on \textit{all graphs} since
all require an update of an APSP problem.
 For this problem, no algorithm exists
which has better worst-case running time than recomputation~\cite{DBLP:journals/algorithmica/RodittyZ11}. In addition, the problem of updating BC
seems even harder than the dynamic APSP problem. Indeed,
the dependencies (and therefore BC) might
need to be updated even on nodes whose distance from the source
has not changed, as they could be part of new shortest paths or not
be part of old shortest paths any more.

\paragraph*{Batch dynamic SSSP algorithms.} Dynamic SSSP algorithms recompute distances from a source node after a single edge update or a batch of edge updates.
Algorithms for the batch problem have been published~\cite{Ramalingam92anincremental,Frigioni_semi-dynamicalgorithms,DBLP:conf/wea/BauerW09} and compared in experimental studies~\cite{DBLP:conf/wea/BauerW09,DBLP:conf/wea/DAndreaDFLP14}.
The experiments show that the tuned algorithm \textsf{T-SWSF}~\cite{DBLP:conf/wea/BauerW09} performs well on many types of graphs and edge updates.
Therefore we use \textsf{T-SWSF} as a building block in our fully-dynamic BC approximation algorithm.

%\newpage

\section{\textsf{RK} algorithm}
\label{sec:rk}
In this section we briefly describe the static BC approximation algorithm by Riondato and Kornaropoulos \textsf{(RK)}~\cite{DBLP:conf/wsdm/RiondatoK14}, 
the foundation for our incremental approach. The idea of \textsf{RK} is to sample a set $S =\{p_{(1)},...,p_{(r)}\}$ of $r$ shortest paths between randomly sampled source-target pairs $(s, t)$.
Then, \textsf{RK} computes the approximated
betweenness centrality $\tilde{c}_B(v)$ of a node $v$ as the
fraction of sampled paths $p_{(k)}\in S$ that $v$ is internal
to, by adding $\frac{1}{r}$ to the node's score for each of these paths.
Figure~\ref{fig:rk-idea} illustrates an example where the sampling of two shortest paths leads to $\frac{2}{r}$ and $\frac{1}{r}$ being added to the score of $u$ and $v$, respectively.
Each possible shortest path $p_{st}$ has the following probability of being sampled in each of the $r$ iterations:
$\pi_{G}(p_{st})=\frac{1}{n(n-1)}\cdot\frac{1}{\sigma_{st}}$ (Lemma 7 of~\cite{DBLP:conf/wsdm/RiondatoK14}).
\begin{wrapfigure}{r}{0.3\textwidth}
  \begin{center}
    \includegraphics[width=0.25\textwidth]{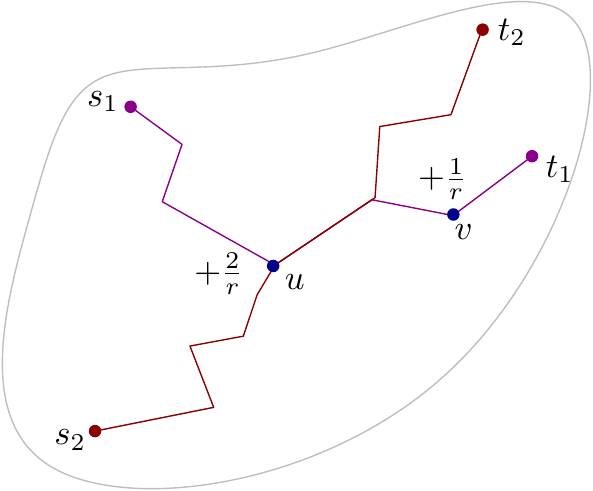}
  \end{center}
  \caption{Sampled paths and score update in the \textsf{RK} algorithm}
  \label{fig:rk-idea}
\end{wrapfigure}
\noindent The number $r$ of samples required to approximate BC scores with the given error guarantee is calculated as

\begin{equation}
r=\frac{c}{\epsilon^{2}}\left(\lfloor\log_{2}\left(\mathit{VD}-2\right)\rfloor+1+\ln\frac{1}{\delta}\right),
\label{sample_size}
\end{equation}

\noindent where $\epsilon$ and $\delta$ are constants in $(0,1)$, $c \approx 0.5$ and \vd is the vertex diameter of $G$, i.e.\ the
number of nodes in the shortest path of $G$ with maximum number of nodes. In unweighted graphs \vd coincides with \textsf{diam}$+1$,
where \textsf{diam} is the number of edges in the longest
shortest path. In weighted graphs \vd and the (weighted) diameter \textsf{diam} (\ie the length of the longest shortest path)
are unrelated quantities. 
The following error guarantee holds:

\begin{lemma}
\textup{\cite{DBLP:conf/wsdm/RiondatoK14}} \label{lem:1} If $r$ shortest paths
are sampled according to the above-defined probability distribution
$\pi_{G}$, then with probability at least $1-\delta$ the approximations
$\tilde{c}_B(v)$ of the betweenness centralities are within $\epsilon$
from their exact value:
$
\Pr(\exists v\in V\: s.t.\:|c_{B}(v)-\tilde{c}_B(v)|>\epsilon)<\delta.
$

\end{lemma}
To sample the shortest paths according to $\pi_{G}$, \textsf{RK} first chooses
a node pair $(s,t)$ uniformly at random and performs an SSSP search from $s$, keeping also track of the number $\sigma_{sv}$
of shortest paths from $s$ to $v$ and of the list of predecessors
$P_{s}(v)$ for any node $v$. Then one shortest path is selected: 
Starting from $t$, a predecessor $z\in P_{s}(t)$
is selected with probability $\sigma_{sz}/\sum_{w\in P_{s}(t)}\sigma_{sw}=\sigma_{sz}/\sigma_{st}$.
The sampling is repeated iteratively
until node $s$ is reached. Algorithm \ref{algo:rk} is the
pseudocode for \textsf{RK}. Function \texttt{computeExtendedSSSP} is 
an SSSP algorithm that keeps track of the number of shortest paths and
of the list of predecessors while computing distances, as in \textsf{BA}.
\begin{algorithm2e}
 \begin{small}
\LinesNumbered
\SetKwData{B}{$\tilde{c}_B$}\SetKwData{VD}{VD}
\SetKwFunction{getVertexDiameter}{getVertexDiameter}
\SetKwFunction{sampleUniformNodePair}{sampleUniformNodePair}
\SetKwFunction{computeExtendedSSSP}{computeExtendedSSSP}
\SetKwInOut{Input}{Input}\SetKwInOut{Output}{Output}
\Input{Graph $G=(V,E),\epsilon,\delta \in (0,1)$}
\Output{Approximated BC values $\forall v\in V$}
\ForEach{node $v \in V$}
{
	\B$(v)\leftarrow 0$\;
}
\VD($G$)$\leftarrow$\getVertexDiameter{$G$}\;
$r \leftarrow (c/\epsilon^2) (\lfloor \log_2(\VD(G)-2)\rfloor +\ln(1/\delta))$\;
\For{$i \leftarrow 1$ \KwTo $r$}{
	$(s_i,t_i)\leftarrow$ \sampleUniformNodePair{$V$}\;
	\mbox{$(d_{s_i},\sigma_{s_i},P_{s_i})\leftarrow$ \computeExtendedSSSP{$G,s_i$}}\;
	\tcp{Now one path from $s_i$ to $t_i$ is sampled uniformly at random}
	$v \leftarrow t_i$\;
	$p_{(i)}\leftarrow$ empty list\;
	\While{$P_{s_i}(v) \neq \{s_i\}$}
	{
		\mbox{sample $z \in P_{s_i}(v)$ with $\Pr=\sigma_{s_i}(z)/\sigma_{s_i}(v)$}\;
		$\B(z)\leftarrow \B(z)+1/r$\;
		add $z\rightarrow p_{(i)}$;
		$v\leftarrow z$\;
	}
}
\Return{$\{(v,\B(v)),\: v\in V\}$}
\end{small}
\caption{\rk algorithm}
\label{algo:rk}
\end{algorithm2e}
\paragraph{Approximating the vertex diameter.}
\rk uses two upper bounds on \vd that can be both computed in $O(n+m)$. For unweighted undirected graphs, it samples a source node $s_i$ for each connected component of $G$, computes a BFS from each $s_i$ and sums the two shortest paths with maximum length starting in $s_i$. The \vd approximation is the maximum of these sums over all components. For directed or weighted graphs, \rk approximates \vd with the size of the largest weakly connected component, which can be a significant overestimation for complex networks, possibly of orders of magnitude. In this paper, we present new approximations for directed and for weighted graphs, described in Section~\ref{sec:new_vd_approx}.
\section{New upper bounds on the vertex diameter}
\label{sec:new_vd_approx}
\subsection{Directed unweighted graphs.}
\label{sub:directed}
Let $G$ be a directed unweighted graph. For now, let us assume $G$ is strongly connected. Let $s$ be any node in $G$ and let $u$ be the node with maximum forward distance from $s$ (\ie $d(s, u) \geq d(s, x)\  \forall x\in V$). Analogously, let $v$ be the node with maximum backward distance (\ie $d(v, s) \geq d(x, s)\  \forall x\in V$). Then, naming $\tilde{\mathit{VD}}_{\text{SC}}$ the sum $d(s,u)+d(v,s)+1$:
\begin{lemma}
\label{lemma:vd_directed}
$\mathit{VD} \leq \tilde{\mathit{VD}}_{\text{SC}} < 2 \mathit{VD} $. 
\end{lemma}
\begin{proof}
Let $x$ and $y$ be two nodes such that the number of nodes in the shortest path from $x$ to $y$ is equal to \vd. Due to the triangle inequality, $d(x,y) \leq d(x,s)+d(s,y)$. Therefore, $d(x,y) \leq d(v,s)+d(s,u)$. Since in unweighted graphs $d(x,y) = \mathit{VD} -1$, the first inequality holds. By definition of \vd, $2\cdot \mathit{VD} \geq (d(v,s)+1) + (d(s,u)+1) > \tilde{\mathit{VD}}_{\text{SC}}$.
 \qed
\end{proof}
The upper bound $\tilde{\mathit{VD}}_{\text{SC}}$ can be computed in $O(n+m)$, simply by running a forward and a backward BFS from any source node $s$.

Let us now consider any directed unweighted graph $G$. We can define the directed acyclic graph $\mathcal{G} = (\mathcal{V}, \mathcal{E})$ of strongly connected components (SCCs) (sometimes referred to as \emph{shrunken graph} in the literature)
similarly to Borassi \etal~\cite{DBLP:journals/corr/BorassiCM15}. In $\mathcal{G}$, the set of vertices is the set of SCCs of $G$ and there is an edge from $C \in \mathcal{V}$ to $C' \in \mathcal{V}$ if and only if there is an edge in $E$ from a node in $C$ to a node in $C'$. (Notice that this also means that all the nodes in $C'$ are reachable from the nodes in $C$.) More in general, the set of nodes reachable from any node in $C$ is exactly the set of nodes in the SCCs reachable from $C$ in $\mathcal{G}$.

We can now define an algorithm that computes an upper bound on \vd for $G$. For each $C$ in $\mathcal{V}$, we compute an upper bound $\tilde{\mathit{VD}}_{\text{SC}}(C)$ on \vd in $C$ (\ie, considering only paths that are contained in $C$) as described before for strongly connected graphs. This can be done in linear time by running a backward and a forward BFS from a random source node $s$ for each SCC and stopping the search when a visited node belongs to a different SCC from that of $s$. For each $C_i$, we know that the nodes reachable from nodes in $C_i$ are only those in the SCCs $C_j$ such that there is a path in $\mathcal{G}$ from $C_i$ to $C_j$. We can compute a topological sorting $\{C_1, \ldots, C_k \}$ of $\mathcal{V}$ (that is, $(C_i,C_j) \in \mathcal{E} \Rightarrow i<j$). Let $C_k$ be the last component in the topological ordering. Then, we know that no path from a node in $C_k$ to any node that is not in $C_k$ exists, which means that the node $C_k$ in $\mathcal{G}$ has outdegree 0.

We call $\tilde{\mathit{VD}}_{\text{DIR}}(C)$ the upper bound on \vd restricted only to nodes that \textit{start} in $C$ (but may end somewhere else). For $C_k$, $\tilde{\mathit{VD}}_{\text{DIR}}(C_k)$ is equal to $\tilde{\mathit{VD}}_{\text{SC}}(C_k)$. For the other components, we can compute it in the following way: Starting from $C_k$, we process all the components in reverse topological ordering and set
$$\tilde{\mathit{VD}}_{\text{DIR}}(C) = \max_{(C,C')\in \mathcal{E}}\tilde{\mathit{VD}}_{\text{DIR}}(C') + \tilde{\mathit{VD}}_{\text{SC}}(C).$$
To get an upper bound on the whole graph, we can take the maximum over all upper bounds $\tilde{\mathit{VD}}_{\text{DIR}}(C)$, \ie we set $\tilde{\mathit{VD}}_{\text{DIR}} := \max_{C \in \mathcal{V}} \tilde{\mathit{VD}}_{\text{DIR}}(C)$. In other words:
$$\tilde{\mathit{VD}}_{\text{DIR}} = \max_{p \in \mathcal{P(G)}} \sum_{C_i \in p} \tilde{\mathit{VD}}_{\text{SC}}(C_i)$$ where $\mathcal{P}(G)$ is the set of paths in $\mathcal{G}$.
\begin{proposition}
\label{prop:vd_directed}
$\mathit{VD} \leq \tilde{\mathit{VD}}_{\text{DIR}} < 2 \cdot \mathit{VD}^2 $. 
\end{proposition}
\begin{proof}
Let us prove the first inequality. Let $p = (u_1, \ldots, u_{\mathit{VD}})$ be a shortest path in $G$ whose number $|p|$ of nodes is equal to \vd. Say $p$ traverses $l$ SCCs $(C_1, ..., C_l)$. Then $p$ can be partitioned in $l$ subpaths $p_i, \ i=1,..,l$, such that $p_i \subseteq C_i$ and $p_i$ is a shortest path in $C_i$. By 
Lemma~\ref{lemma:vd_directed}, $|p_i| \leq \tilde{\mathit{VD}}_{\text{SC}}(C_i),\ \ i=1,...,l$. Therefore, $|p| = \sum_{i=1,...,l} |p_i| \leq  \sum_{i=1,...,l} \tilde{\mathit{VD}}_{\text{SC}}(C_i) \leq \tilde{\mathit{VD}}_{\text{DIR}}$ (this last inequality holds by definition of $\tilde{\mathit{VD}}_{\text{DIR}}$).

How ``bad'' can $\tilde{\mathit{VD}}_{\text{DIR}}$ be in the worst case? Let now $(C_1, ..., C_l)$ denote the path in $\mathcal{G}$ such that $\tilde{\mathit{VD}}_{\text{DIR}} =  \sum_{i=1,..,l} \tilde{\mathit{VD}}_{\text{SC}}(C_i)$. Let $l$ be the number of components in this path. How large can $l$ be?  Since there is at least one node of $G$ in each $C_i$, there must be at least a shortest path of size $l$ in $G$ that goes through the components $C_1,...,C_l$. Therefore, $l \leq \mathit{VD}$. Also, let $k=\max_{C \in \mathcal{V}} \tilde{\mathit{VD}}_{\text{SC}}(C)$. By Lemma~\ref{lemma:vd_directed}, $k < 2\cdot \mathit{VD}(C_k)$, where $C_k$ is the SCC whose upper bound is equal to $k$. Clearly, $k <2\cdot \mathit{VD}$. Then, by definition, $\tilde{\mathit{VD}}_{\text{DIR}} = \sum_{i=1,..,l} \tilde{\mathit{VD}}_{\text{SC}}(C_i) \leq l \cdot k < 2 \cdot \mathit{VD}^2$.
 \qed
\end{proof}
The upper bound can be computed in $O(n+m)$. Indeed, $\mathcal{G}$ can be computed in $O(n+m)$, by finding the SCCs of $G$ and scanning the edges in $E$. Then, the topological sorting and the accumulation of the $\tilde{\mathit{VD}}_{\text{DIR}}(C)$ of the different components can be done in $O(|\mathcal{V}|+|\mathcal{E}|) = O(n+m)$.
\begin{figure}[htb]
  \begin{center}
  %\vspace{-4ex}
    \includegraphics[width=0.45\textwidth]{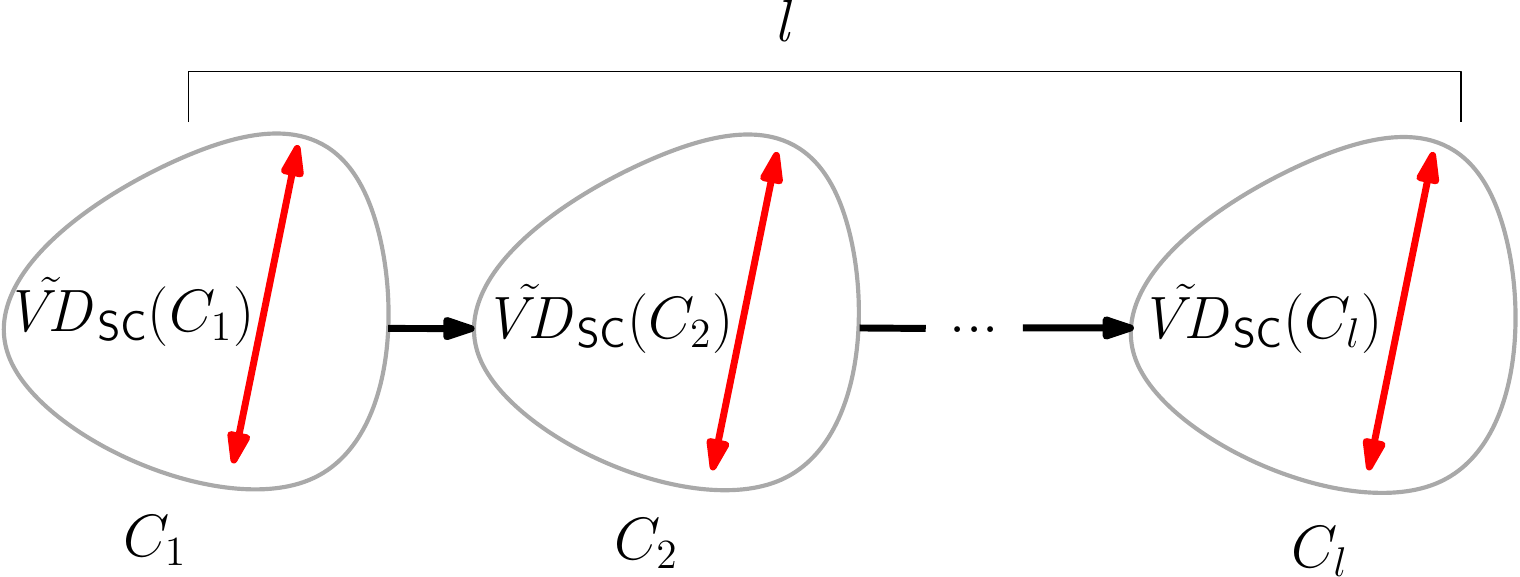}
  \end{center}
  \caption{Path $(C_1, ..., C_l)$ of the DAG $\mathcal{G}$ of SCCs. Each SCC $C_i$ has its own upper bound $\tilde{\mathit{VD}}_{\text{SC}}(C_i)$ and $\tilde{\mathit{VD}}_{\text{DIR}}$ is computed as $\sum_{i=1,..,l} \tilde{\mathit{VD}}_{\text{SC}}(C_i)$.}
  \label{fig:vd_directed}
  %\vspace{-2ex}
\end{figure}
Notice that our new upper bound is never larger than the size of the largest weakly connected component, the previous bound used in \rk. Also, although the upper bound can be as bad as $2\cdot \mathit{VD}^2$ in theory, our experimental results on real-world complex networks show that it is within a factor 4 from the exact \vd and several orders of magnitude smaller than the size of the largest weakly connected components.

\subsection{Undirected weighted graphs}
\label{sub:weighted}
Let $G$ be an undirected graph. For simplicity, let $G$ be connected for now. If it is not, we compute an approximation for each connected component and take the maximum over all the approximations. Let $T \subseteq G$ be an SSSP tree from any source node $s \in V$. Let $p_{xy}$ denote a shortest path between $x$ and $y$ in $G$ and let $p_{xy}^T$ denote a shortest path between $x$ and $y$ in $T$. Let $|p_{xy}|$ be the number of nodes in $p_{xy}$ and $d(x,y)$ be the distance between $x$ and $y$ in $G$, and analogously for $|p_{xy}^T|$ and $d^T(x,y)$. Let $\overline{\omega}$ and  $\underline{\omega}$ be the maximum and minimum edge weights in $G$, respectively. Let $u$ and $v$ be the nodes with maximum distance from $s$, \ie $d(s, u)\geq d(s,v) \geq d(s, x)\  \forall x\in V, x\neq u $. 
We define the \vd approximation $\tilde{\mathit{VD}}_{\text{W}} := 1 + \frac{d(s,u)+d(s,v)}{\underline{\omega}}$. Then:
\begin{proposition}
\label{lem:vd2}
$\mathit{VD} \leq \tilde{\mathit{VD}}_{\text{W}} < 2 \cdot \mathit{VD} \cdot \overline{\omega} / \underline{\omega}$. 
\end{proposition}
\begin{proof}
To prove the first inequality, we can notice that $d^T(x,y) \geq d(x,y)$ for all $x,y \in V$, since all the edges of $T$ are contained in those of $G$. Also, since every edge has weight at least $\underline{\omega}$, $d(x,y) \geq (|p_{xy}|-1)\cdot \underline{\omega}$. Therefore, $d^T(x,y) \geq (|p_{xy}|-1)\cdot \underline{\omega}$, which can be rewritten as $|p_{xy}| \leq 1+ \frac{d^T(x,y)}{\underline{\omega}}$, for all $x,y \in V$. Thus, $\mathit{VD} = \max_{x,y} |p_{xy}| \leq 1+(\max_{x,y} d^T(x,y))/ \underline{\omega} \leq 1 + \frac{d^T(s,u)+d^T(s,v)}{\underline{\omega}} = 1 + \frac{d(s,u)+d(s,v)}{\underline{\omega}}$, where the last expression equals $\tilde{\mathit{VD}}_{\text{W}}$ by definition.

To prove the second inequality, we first notice that $d(s,u) \leq (|p_{su}|-1)\cdot \overline{\omega}$,
 and analogously $d(s,v) \leq (|p_{sv}|-1)\cdot \overline{\omega}$. Consequently, $\tilde{\mathit{VD}}_{\text{W}} \leq 1+(|p_{su}| + |p_{sv}|-2)\cdot \overline{\omega} / \underline{\omega} < 2 \cdot |p_{su}|\cdot \overline{\omega} / \underline{\omega}$, supposing that $|p_{su}| \geq |p_{sv}|$ without loss of generality. By definition of \vd, $|p_{su}| \leq \mathit{VD}$. Therefore, $\tilde{\mathit{VD}}_{\text{W}} < 2 \cdot \mathit{VD} \cdot \overline{\omega} / \underline{\omega}$.
 \qed
\end{proof}
To obtain the upper bound $\tilde{\mathit{VD}}_{\text{W}}$, we can simply compute an SSSP search from any node $s$, find the two nodes with maximum distance and perform the remaining calculations.
Notice that $\tilde{\mathit{VD}}_{\text{W}}$ extends the upper bound proposed for \rk~\cite{DBLP:conf/wsdm/RiondatoK14} for unweighted graphs: When the graph is unweighted and thus $\underline{\omega} = \overline{\omega}$, $\tilde{\mathit{VD}}_{\text{W}}$ becomes equal to the approximation used by \rk.
Complex networks are often characterized by a small diameter and in networks like coauthorship, friendship, communication networks, \vd and $\overline{\omega} / \underline{\omega}$ can be several order of magnitude smaller than the size of the largest component, though. In this case the new bound translates into a substantially improved \vd approximation.

\subsection{Directed weighted graphs.}
\label{sub:dir_weighted}
The upper bound for directed weighted graphs can be easily derived from those described in Sections~\ref{sub:directed} and~\ref{sub:weighted}. If $G$ is strongly connected, we can define $\tilde{\mathit{VD}}_{\text{SCW}} := 1 + \frac{d(s,u)+d(v,s)}{\underline{\omega}}$, where $s$ is any node,  $u$ is the node with maximum forward distance from $s$, $v$ is the node with maximum backward distance and $\underline{\omega}$ is the minimum edge weight. It can be easily proved that Proposition~\ref{lem:vd2} holds also in this case, considering a forward SSSP tree from $s$ (where distances are bounded by $d(s,u)$) and a backward SSSP tree (where distances are bounded by $d(v,s)$). For general directed weighted graphs, we can use the algorithm described in Section~\ref{sub:directed} using $\tilde{\mathit{VD}}_{\text{SCW}}(C) := 1 + \frac{d(s,u)+d(v,s)}{\underline{\omega}_c}$ as an upper bound for each SCC $C$ (where $\underline{\omega}_c$ is the minimum edge weight in $C$). It is easy to prove that the resulting bound is an upper bound on \vd and that it is always smaller than $2\cdot \max_{C \in \mathcal{G}} \frac{\overline{\omega}_C}{\underline{\omega}_C} \cdot \mathit{VD} ^2$, using Propositions~\ref{prop:vd_directed} and~\ref{lem:vd2}.
Since it requires to compute an SSSP tree for each SCC, the complexity of computing the bound is $O(m + n \log n)$ (the other operations can be done in linear time, as described in Section~\ref{sub:directed}).

\section{Fully-dynamic approximation algorithms}
\label{sec:dyn-algo}
\subsection{Path subsitution.}
\label{sub:substitution}
\begin{wrapfigure}{r}{0.4\textwidth}
  \begin{center}
%  \vspace{-4ex}
    \includegraphics[width=0.3\textwidth]{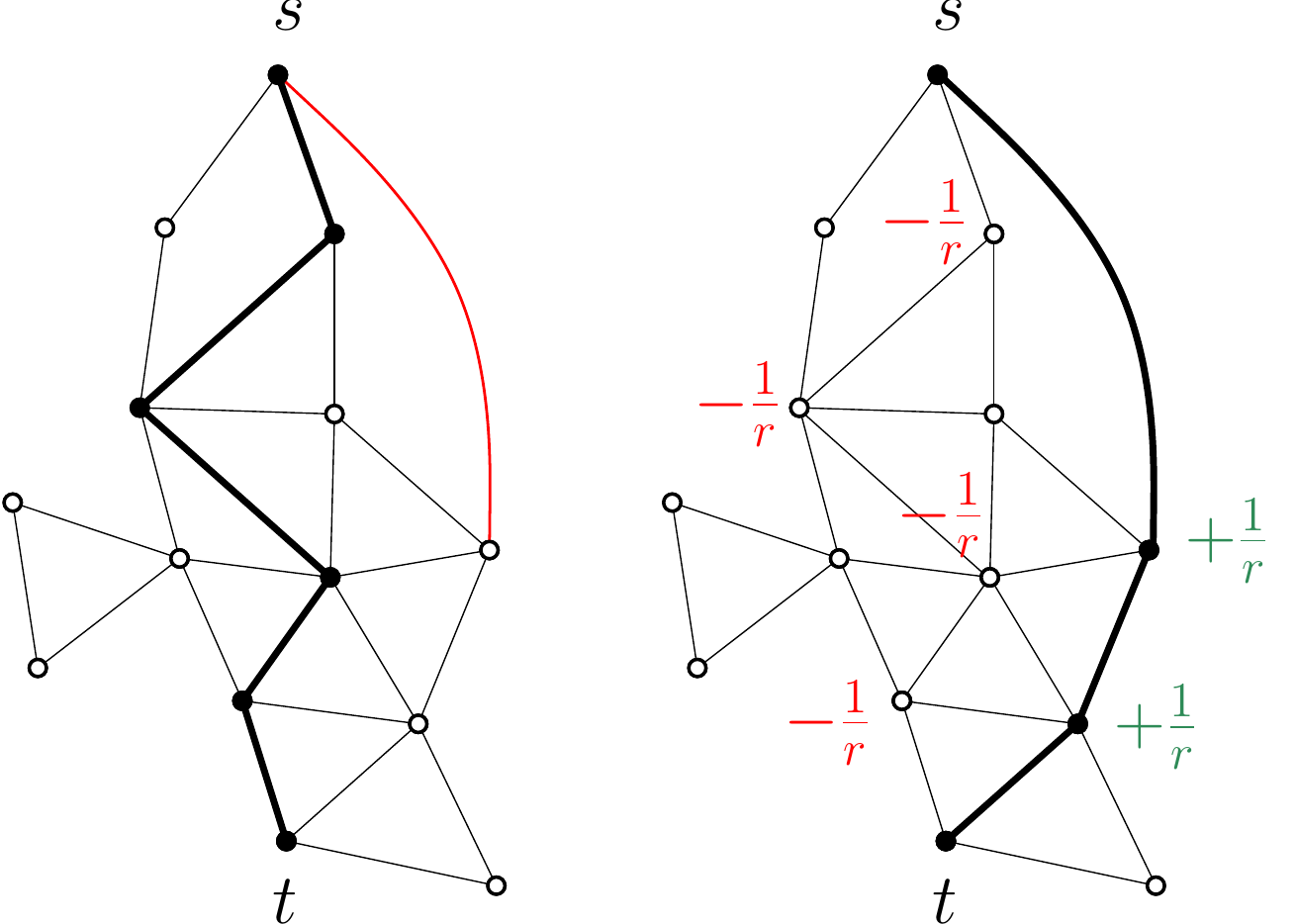}
  \end{center}
  \caption{Updating shortest paths and BC scores}
  \label{fig:path-update}
%  \vspace{-2ex}
\end{wrapfigure}
Our algorithms for dynamic BC approximation are composed of two phases: an \textit{initialization} phase, which executes \textsf{RK} on the initial graph, and an \textit{update} phase, which recomputes the approximated BC scores after a sequence of edge updates. 
Let us consider a batch $\beta = \{e_1,...,e_k\}$ of edge updates $e_i=\{u_i,v_i, \omega(u_i, v_i)\}$ applied to a graph $G$. Also, let us assume for the moment that $\beta$ is composed of \textbf{edge insertions only} (or weight decreases) and $\beta$ \textbf{does not increase} the vertex diameter of $G$ and therefore also the number $r$ of samples required by \rk for the maximum error guarantee. We will discuss the general case in Section~\ref{sub:fully-dyn}.
Intuitively, our basic idea is to keep the old sampled paths and update them only when necessary, instead of re-computing $r$ shortest paths from scratch.
Figure~\ref{fig:path-update} shows an example to illustrate this idea: Assume several shortest paths between $s$ and $t$ exist, of which one has been sampled (with black nodes). An edge insertion (represented in red) shortens the distance between $s$ and $t$, creating a new shorter path. Therefore, we simply subtract $1/r$ from each node in the old shortest path and add $1/r$ to each node in the new one.

From this point on, we give a formal description and only consider edge insertions. We suppose the graph is undirected, but
in this restricted semi-dynamic setting our results can be easily extended to weight decreases and directed graphs. 
Let $G'=(V,E\cup \beta)$ be the new graph,
let $d'_{s}(t)$ denote the new distance between any node pair $(s,t)$
and let $\sigma'_{st}$ be the new number of shortest paths between $s$ and $t$.
Let $\mathcal{S}_{st}$ and $\mathcal{S}'_{st}$ be the old and the new set of shortest paths between $s$ and $t$, respectively.
A new set $ S'=\{p'_{(1)},...,p'_{(r)}\} $
of shortest paths has to be sampled now in order to let Lemma~\ref{lem:1}
hold for the new configuration; in particular, the probability  $\Pr(p'_{(k)}=p'_{st})$ of
each shortest path $p'_{st}$ to be sampled
must be equal to $\pi_{G'}(p'_{st})=\frac{1}{n(n-1)}\cdot\frac{1}{\sigma'_{st}}$.
Clearly, one could rerun \textsf{RK} on the new graph, but we can be more efficient:
 We recall that we are assuming that the \vd and therefore also the number of samples for $G'$ is smaller than or equal to that of $G$.
Given any old sampled path $p_{st}$, we can \emph{keep}
$p_{st}$ if the set of shortest paths between $s$ and $t$ has not
changed, and \emph{replace} it with a new path between $s$ and
$t$ otherwise. Then, the following lemma holds:

\begin{lemma}
\label{lem:2}Let $ S$ be a set of shortest paths of $G$
sampled according to $\pi_{G}$. Let $\mathcal{P}$ be the procedure
that creates $ S'$ by substituting each path $p_{st}\in S$
with a path $p'_{st}$ according to the following rules:
1. $p'_{st}=p_{st}$ if $d'_{s}(t)=d_{s}(t)$ and $\sigma'_{st}=\sigma_{st}$.
2. $p'_{st}$ selected uniformly at random among $\mathcal{S}'_{st}$
\mbox{otherwise}.
Then, $p'_{st}$ is a shortest path of $G'$ and the probability of
any shortest path $p'_{xy}$ of $G'$ to be sampled at each iteration is
$\pi_{G'}(p'_{xy})$, i.e.\
$\Pr(p'_{(k)}=p'_{xy})=\frac{1}{n(n-1)}\cdot\frac{1}{\sigma'_{xy}}, k=1,...,r\text{.}$
\end{lemma}
\begin{proof}
To see that $p'_{st}$ is a shortest path of $G'$, it is sufficient
to notice that, if $d'_{s}(t)=d_{s}(t)$ and $\sigma'_{s}(t)=\sigma_{s}(t)$,
then all the shortest paths between $s$ and $t$ in $G$ are shortest
paths also in $G'$.

Let $p'_{xy}$ be a shortest path of $G'$ between nodes $x$ and
$y$. Basically, there are two possibilities for $p'_{xy}$ to be the $k$-th sample. Naming $e_{1}$ the event $\{\mathcal{S}_{xy}=\mathcal{S}'_{xy}\}$
(the set of shortest paths between $x$ and $y$ does not change after
the edge insertion) and $e_{2}$ the complementary event of $e_{1}$,
we can write $\Pr(p'_{(k)}=p'_{xy})$ as $\Pr(p'_{(k)}=p'_{xy}\cap e_{1})+\Pr(p'_{(k)}=p'_{xy}\cap e_{2})$.

Using conditional probability, the first addend can be rewritten as
$\Pr(p'_{(k)}=p'_{xy}\cap e_{1})=\Pr(p'_{(k)}=p'_{xy}|e_{1})\Pr(e_{1})$.
As the procedure $\mathcal{P}$ keeps the old shortest path when $e_{1}$
occurs, then $\Pr(p'_{(k)}=p'_{xy}|e_{1})=\Pr(p_{(k)}=p'_{xy}|e_{1})=\frac{1}{n(n-1)}\frac{1}{\sigma_{x}(y)}$,
which is also equal to $\frac{1}{n(n-1)}\frac{1}{\sigma'_{x}(y)}$,
since $\sigma_{x}(y)=\sigma'_{x}(y)$ when we condition on $e_{1}$.
Therefore, $\Pr(p'_{(k)}=p'_{xy}\cap e_{1})=\frac{1}{n(n-1)}\frac{1}{\sigma'_{x}(y)}\cdot\Pr(e_{1})$.

Analogously, $\Pr(p'_{(k)}=p'_{xy}\cap e_{2})=\Pr(p'_{(k)}=p'_{xy}|e_{2})\Pr(e_{2})$.
In this case, $\Pr(p'_{(k)}=p'_{xy}|e_{2})=\frac{1}{n(n-1)}\cdot\frac{1}{\sigma'_{x}(y)}$,
since this is the probability of the node pair $(x,y)$ to be the $k$-th sample in the initial sampling and of $p'_{xy}$ to be selected
among other paths in $\mathcal{S}'_{xy}$. Then, $\Pr(p'_{(k)}=p'_{xy}\cap e_{2})=\frac{1}{n(n-1)}\cdot\frac{1}{\sigma'_{x}(y)}\cdot\Pr(e_{2})=\frac{1}{n(n-1)}\cdot\frac{1}{\sigma'_{x}(y)}\cdot(1-\Pr(e_{1}))$.
The probability $\Pr(p'_{(k)}=p'_{xy})$ can therefore be
rewritten as $ \Pr(p'_{(k)}=p'_{xy}) =\mbox{\ensuremath{\frac{1}{n(n-1)}\frac{1}{\sigma'_{x}(y)}\cdot\Pr}(\ensuremath{e_{1}})+\ensuremath{\frac{1}{n(n-1)}\frac{1}{\sigma'_{x}(y)}\cdot}(1-\ensuremath{\Pr}(\ensuremath{e_{1}}))}=\frac{1}{n(n-1)}\frac{1}{\sigma'_{x}(y)}.$
\qed 
\end{proof}
Since the set of paths is constructed according to $\pi_{G'}$, Theorem~\ref{theorem_samples} follows directly from Lemma \ref{lem:1}.
\begin{theorem}
\label{thm:approx_guarantee}
Let $G=(V,E)$ be a graph and let $G'=(V,E\cup \beta)$
be the modified graph after the the batch $\beta$. Let $VD(G)\geq VD(G')$.
Let $ S$ be a set of $r$ shortest paths of $G$
sampled according to $\pi_{G}$
and $r=\frac{c}{\epsilon^{2}}\left(\lfloor\log_{2}\left(\mathit{VD}(G)-2\right)\rfloor+1+\ln\frac{1}{\delta}\right)$
for some constants $\epsilon,\delta\in(0,1)$. Then, if a new set
$ S'$ of shortest paths of $G'$ is built according to procedure
$\mathcal{P}$ and the approximated values of betweenness centrality
$\tilde{c}'_B(v)$ of each node $v$ are computed as the fraction
of paths of $ S'$ that $v$ is internal to, then
\[
\Pr(\exists v\in V\: s.t.\:|c'_{B}(v)-\tilde{c}'_B(v)|>\epsilon)<\delta,
\]
where $c'_{B}(v)$ is the new exact value of betweenness centrality
of $v$ after the edge insertion.
\label{theorem_samples}
\end{theorem}
Algorithm~\ref{algo:main_algo}
shows the update procedure based on
Theorem~\ref{thm:approx_guarantee}. For each sampled node pair $(s_{i},t_{i}),\, i=1,...,r$,
we first update the SSSP from $s_i$, a step which will be
discussed in Sections~\ref{sub:update_weighted} and~\ref{sub:update_unweighted}.
 In case the distance
or the number of shortest paths between $s_{i}$ and $t_{i}$ has
changed, a new shortest path is sampled uniformly as in \textsf{RK}. 
This means that $\frac{1}{r}$ is subtracted from the score of each node in the old shortest path and the same quantity
is added to the nodes in the new shortest
path. On the other hand, if both distances and number of shortest
paths between $s_{i}$ and $t_{i}$ have not changed, nothing needs
to be updated. 
Considering edges in a batch allows us to recompute the BC scores only once instead of doing it after each single edge update. Moreover, this gives us the possibility to use specific batch algorithms for the update of the SSSP DAGs, which process the nodes affected by multiple edges of $\beta$ only once, instead of for each single edge.

Differently from \rk, in our dynamic algorithm we scan the neighbors every time we need the predecessors instead of storing them (Line~\ref{predecessors}). This allows us to use $\Theta(n)$ memory per sample (\ie, $\Theta(r \cdot n)$ in total) instead of $\Theta(m)$ per sample, while our experiments show that the running time is hardly influenced. The number of samples depends on $\epsilon$, so in theory this can be as large as $|V|$. However, our experiments show that relatively large values of $\epsilon$ (\eg $\epsilon=0.05$) lead to a good ranking of nodes with high BC and for such values the number of samples is typically much smaller than $|V|$, making the memory requirements of our algorithms significantly less demanding than those of the dynamic exact algorithms ($\Omega(n^2)$) for many applications.

\begin{algorithm2e}
\begin{small}
\LinesNumbered
\SetKwData{B}{$\tilde{c}_B$}
\SetKwInOut{Input}{Input}\SetKwInOut{Output}{Output}
\Input{Graph $G=(V,E)$, source node $s$, number of iterations $r$, batch $\beta$  of edge insertions/weight decreases}
\Output{New approximated BC values $\forall v\in V$}
\For{$i \leftarrow 1$ \KwTo $r$}
{
	$d^{old}_i\leftarrow d_{s_i}(t_i)$\;
	$\sigma^{old}_i\leftarrow \sigma_{s_i}(t_i)$\;
	$(d_{s_i},\sigma_{s_i})\leftarrow$\textsf{UpdateSSSP}($G,d_{s_i},\sigma_{s_i},\beta$)\; \label{lst:line:update_sssp}
	%\tcp{If the shortest paths between $s_i$ and $t_i$ have changed, we sample a new one}
	\If{$d_{s_i}(t_i)<d^{old}_i$ \textbf{or} $\sigma_{s_i}(t_i)\neq \sigma^{old}_i$}
	{\label{check_paths}
			%\tcp{We subtract $1/r$ from all nodes in the old shortest path and we add $1/r$ to all nodes in the new shortest path}
			\ForEach{$w \in p_{(i)}$}
			{
				\B($w$) $\leftarrow \B(w)-1/r$\;
			}
			$v \leftarrow t_i$\;
			$p_{(i)}\leftarrow$ empty list\;
			$P_{s_i}(v)\leftarrow \{u\in V | (u,v) \in E \wedge d_{s_i}(u)+\omega(u,v) = d_{s_i}(v)\}$\; \label{predecessors}
			\While{$P_{s_i}(v) \neq \{s_i\}$}
			{
				\mbox{sample $z \in P_{s_i}(v)$ with $\Pr=\sigma_{s_i}(z)/\sigma_{s_i}(v)$}\;
				$\B(z)\leftarrow \B(z)+1/r$\;
				add $z$ to $p_{(i)}$\;
				$v\leftarrow z$\;
			}
	}
}
\Return{$\{(v,\B(v)),\: v\in V\}$}
\end{small}
\caption{BC update after a batch $\beta$ of edge insertions/weight decreases}
\label{algo:main_algo}
\end{algorithm2e}

\subsection{Sampling new paths.}
\label{sub:fully-dyn}
In the previous section, we assumed that $\mathit{VD}(G) \geq \mathit{VD}(G')$. Although many real-world networks exhibit a shrinking-diameter behavior~\cite{DBLP:conf/kdd/LeskovecKF05}, it is clearly possible that \vd increases as a consequence of edge insertions/deletions or weight updates. If this happens, we can still update the old paths as described in Section~\ref{sub:substitution}, but we also need to sample new additional paths, according to the probability distribution $\pi_G'$. The general algorithm to update the BC scores after a batch $\beta$ could be described as follows: First, we update the old shortest paths as described in Section~\ref{sub:substitution}. Then, we recompute an upper bound on $\mathit{VD}(G')$ in linear time, using the algorithms described in Section~\ref{sec:new_vd_approx}. Using $\mathit{VD}(G')$, we compute the number of samples $r(G')$ defined in Eq.~(\ref{sample_size}). If $r(G') > r(G)$, we sample new $r(G')-r(G)$ additional paths using \rk. For undirected graphs, we also propose two fully-dynamic algorithms (one for weighted and one for unweighted graphs) to keep track of an upper bound on \vd over time (Section~\ref{sub:dyn_vd_approx}). This saves additional time, allowing for a quick recomputation of the upper bound after the batch instead of recomputing it from scratch.

Notice that, if edge deletions are allowed, it is not sufficient to check whether the distance and the number of shortest paths between two nodes $s$ and $t$ has not changed (Line~\ref{check_paths} of Algorithm~\ref{algo:main_algo}), since they might remain unchanged even if the set of shortest paths is actually different. In this case, we always replace the old shortest path with a new one (we basically remove the \textit{if} statement in Line~\ref{check_paths}).

In the following, we present the fully-dynamic algorithms (for weighted and unweighted graphs) needed to update the shortest paths (\textsf{updateSSSP} in Algorithm~\ref{algo:main_algo}) and the fully-dynamic algorithm that recomputes an upper bound on \vd for undirected graphs. Finally, we show how these algorithms can be combined to obtain an even faster algorithm (than the one described in this section) for dynamic BC approximation in undirected graphs (Section~\ref{sub:combined}).

\subsection{SSSP update in weighted graphs.}
\label{sub:update_weighted}
Our SSSP update is based on \textsf{T-SWSF}~\cite{DBLP:conf/wea/BauerW09}, which recomputes distances from a source node $s$ after a batch $\beta$ of weight updates (or edge insertions/deletions). For our purposes, we need two extensions of \textsf{T-SWSF}: an algorithm that also recomputes the number of shortest paths between $s$ and the other nodes (\sssp) and one that also updates a \vd approximation for the connected component of $s$ (\upvd) (the latter is used in the fully-dynamic \vd approximation for undirected graphs, described in Section~\ref{sub:dyn_vd_approx}). 
The \vd approximation is computed as described in Sections~\ref{sec:rk} and~\ref{sub:weighted}. Thus, \upvd keeps track of the two maximum distances $d'$ and $d''$ from $s$ and the minimum edge weight $\underline{\omega}$.

We call \textit{affected nodes} the nodes whose distance (or also whose number of shortest paths, in \sssp) from $s$ has changed as a consequence of $\beta$. Basically, the idea is to put the set $A$ of affected nodes $w$ into a priority queue $Q$ with priority $p(w)$ equal to the candidate distance of $w$. When $w$ is extracted, if there is actually a path of length $p(w)$ from $s$ to $w$, the new distance of $w$ is set to $p(w)$, otherwise $w$ is reinserted into $Q$ with a higher candidate distance. In both cases, the affected neighbors of $w$ are inserted into $Q$. 

Algorithm~\ref{algo:weighted} describes the SSSP update for weighted undirected graphs. The extension to directed graphs is trivial. The pseudocode updates both the \vd approximation for the connected component of $s$ and the number of shortest paths from $s$, so it basically includes both \sssp and \upvd.
Initially, we scan the edges $e =\{u,v\}$ in $\beta$ and, for each $e$, we insert the endpoint with greater distance from $s$ into $Q$ (w.l.o.g., let $v$ be such endpoint). The priority $p(v)$ of $v$ represents the \textit{candidate} new distance of $v$. This is the minimum between the $d(v)$ and $d(u)$ plus the weight of the edge $\{u, v\}$. Notice that we use the expression ``insert $v$ into $Q$'' for simplicity, but this can also mean update $p(v)$ if $v$ is already in $Q$ and the new priority is smaller than $p(v)$. 
When we extract a node $w$ from $Q$, we have two possibilities: (i) there is a path of length $p(w)$ and $p(w)$ is actually the new distance or (ii) there is no path of length $p(w)$ and the new distance is greater than $p(w)$. In the first case (Lines~\ref{update_weighted:part1-1}~-~\ref{update_weighted:part1-2}), we set $d(w)$ to $p(w)$ and insert the neighbors $z$ of $w$ such that $d(z)>d(w)+\omega(\{w,z\})$ into $Q$ (to check if new shorter paths to $z$ that go through $w$ exist). In the second case (Lines~\ref{update_weighted:part2-1}~-~\ref{update_weighted:part2-2}), there is no shortest path between $s$ and $w$ known anymore, so that we set $d(w)$ to $\infty$. We compute $p(w)$ as $\min_{\{v,w\}\in E} d(v) + \omega{(v,w)}$ (the new candidate distance for $w$) and insert $w$ into $Q$. Also its neighbors could have lost one (or all of) their old shortest paths, so we insert them into $Q$ as well. The update of $\underline{\omega}$ can be done while scanning the batch and of $d'$ and $d''$ when we update $d(w)$. 

The pseudocode also updates a global variable $vis(w)$ that keeps track, for each node $w$, of the number of source nodes from which $w$ is reachable. This is necessary for the fully-dynamic \vd approximation and will be explained in Section~\ref{sub:dyn_vd_approx}.
In particular, we decrease $vis(w)$ when updating $d(w)$ in case the old $d(w)$ was equal to $\infty$ (\ie w has become reachable) and we decrease $vis(w)$ when we set $d(w)$ to $\infty$ (\ie $w$ has become unreachable). We update the number of shortest paths after updating $d(w)$, as the sum of the shortest paths of the predecessors of $w$ (Lines~\ref{update_weighted:sp1}~-~\ref{update_weighted:sp2}).
In \upvd, $d'$ and $d''$ are recomputed while updating the distances (Line~\ref{line:update_dd}) and $\underline{\omega}$ is updated while scanning $\beta$ (Line~\ref{line:omega}).  
Let $|\beta |$ represent the cardinality of $\beta$ and let $||A ||$ represent the sum of the nodes in $A$ and of the edges that have at least one endpoint in $A$. Then, the following complexity derives from feeding $Q$ with the batch and inserting into/extracting from $Q$ the affected nodes and their neighbors.
\begin{algorithm2e}
\begin{footnotesize}
\LinesNumbered
\SetKwFunction{insertOrDecreaseKey}{insertOrDecreaseKey}
\SetKwFunction{extractMin}{extractMin}
\SetKwInOut{Input}{Input}\SetKwInOut{Output}{Output}
\Input{Graph $G=(V,E)$, vector $d$ of distances from $s$, vector $\sigma$ of number of shortest paths from $s$, batch $\beta$}
\Output{New values of $d(v)$ and $\sigma(v)$ $\forall v \in V$}
$Q\leftarrow$ empty priority queue\;
\ForEach{$e=\{u,v\} \in \beta, d(u)<d(v)$}
{\label{update_weighted:init1}
	$Q \leftarrow$ \insertOrDecreaseKey{$v, p(v)=\min\{d(u)+\omega(\{u,v\}),d(w)\}$}\;
}\label{update_weighted:init2}
$\underline{\omega} \leftarrow \min\{\underline{\omega},\  \omega(e): e \in \beta\}$\; \label{line:omega}
\While{there are nodes in $Q$}
{
		$\{w, p(w)\} \leftarrow$ \extractMin{$Q$}\;
		$con(w) \leftarrow \min_{z : (z,w)\in E} d(z) + \omega(z,w)$\;
		\If{$con(w) = p(w)$}
		{\label{update_weighted:part1-1}
			update $d'$ and $d''$\; \label{line:update_dd}
			\If{$d(w) = \infty$} 
			{
				$vis(w) \leftarrow vis(w)+1$\; \label{update_weighted:visinc}
			}
			$d(w) \leftarrow p(w)$; $\sigma(w)\leftarrow 0$\;
			\ForEach{incident edge $(z,w)$}
			{
				\If{$d(w)=d(z)+\omega(z,w)$}
				{\label{update_weighted:sp1}
					$\sigma(w)\leftarrow \sigma(w)+\sigma(z)$\; \label{line:sigma}
				}\label{update_weighted:sp2}
				\If{$d(z)\geq d(w)+\omega(z,w)$}
				{\label{update_w:neigh1}
					$Q \leftarrow $ \insertOrDecreaseKey{$z, p(z)=d(w)+\omega(z,w)$}\;
				}\label{update_w:neigh2}
			}
		}\label{update_weighted:part1-2}
		\Else
		{\label{update_weighted:part2-1}
			\If{$d(w) \neq \infty$}
			{
				$vis(w) \leftarrow vis(w)-1$\; \label{update_weighted:visdec}
				\If{vis(w)=0}{
				insert $w$ into $U$\;
				}
				\If{$con(w) \neq \infty$}
				{
					$Q \leftarrow $\insertOrDecreaseKey{$w, p(w)=con(w)$}\;
					\ForEach{incident edge $(z,w)$}
					{
						\If{$d(z)=d(w)+\omega(w,z)$}
						{
							$Q \leftarrow $\insertOrDecreaseKey{$z,p(z)=d(w)+\omega(z,w)$}\;
						}
					}
					$d(w)\leftarrow \infty$\;
				}
			}
		}\label{update_weighted:part2-2}
}
\end{footnotesize} 
\caption{SSSP update for weighted graphs (\textsf{updateSSSP-W})}
\label{algo:weighted}
\end{algorithm2e}

\begin{lemma}
\label{thm:complexity}
The time required by \upvd (\sssp) to update the distances and \vda (the number of shortest paths) is $O(|\beta |\log |\beta | + ||A || \log ||A ||)$.
\end{lemma}
\begin{proof}
In the initial scan of the batch (Lines~\ref{update_weighted:init1}-\ref{update_weighted:init2}), we scan the nodes of the batch and insert the affected nodes into $Q$ (or update their value). This requires at most one heap operation (insert or decrease-key) for each element of $\beta$, therefore $O(|\beta| \log|\beta|)$ time.
When we extract a node $w$ from $Q$, we have two possibilities: (i) $con(w)=p(w)$ (Lines~\ref{update_weighted:part1-1}~-~\ref{update_weighted:part1-2}) or (ii) $con(w)>p(w)$ (Lines~\ref{update_weighted:part2-1}~-~\ref{update_weighted:part2-2}). In the first case, we scan the neighbors of $w$ and perform at most one heap operation for each of them (Lines~\ref{update_w:neigh1}~-~\ref{update_w:neigh2}). In the second case, we scan only if $d(w)\neq \infty$. Therefore, we can perform up to one heap operation per incident edge of $w$, for each extraction of $w$ in which $d(w)\neq \infty$ or $con(w)=p(w)$.
How many times can an affected node $w$ be extracted from $Q$ with $d(w) \neq \infty$ or $con(w)=p(w)$? If the first time we extract $w$, $con(w) $ is equal to $p(w)$ (case (i)), then the final value of $d(w)$ is reached and $w$ is not inserted into $Q$ anymore. If the first time we extract $w$, $con(w)$ is greater than $p(w)$ (case (ii)), then $w$ can be inserted into the queue again. However, its distance is set to $\infty$ and therefore no additional operations are performed, until $d(w)$ becomes less than $\infty$. But this can happen only in case (i), after which $d(w)$ reaches its final value. To summarize, each affected node $w$ can be extracted from $Q$ with $d(w)\neq \infty$ or $con(w)=p(w)$ at most twice and, every time this happens, at most one heap operation per incident edge of $w$ is performed. The complexity is therefore $O(|\beta |\log |\beta | + ||A || \log ||A ||)$. \qed
\end{proof}
\subsection{SSSP update in unweighted graphs.}
\label{sub:update_unweighted}
For unweighted graphs, we basically replace the priority queue $Q$ of \upvd and \sssp with a list of queues. Each queue represents a \emph{level} (distance from $s$) from 0 (which only the source belongs to) to the maximum distance $d'$. The levels replace the priorities and also in this case represent the candidate distances for the nodes. Algorithm~\ref{algo:unweighted} describes the pseudocode for unweighted graphs. As in Algorithm~\ref{algo:weighted}, we
first scan the batch (Lines~\ref{update:batch1} -~\ref{update:batch2}) and insert the nodes
in the queues. Then (Lines~\ref{update:queues1} -~\ref{update:queues2}), we scan the queues in order of increasing distance from $s$, in a fashion similar to that of a priority queue. 
In order not to insert a node in the queues multiple times, we use colors:
Initially we set all the nodes to white and then we set a node $w$ to black only when we find the final distance of $w$ (\ie when we set $d(w)$ to $k$) (Line~\ref{update:black1}). Black nodes extracted from a queue are then skipped (Line~\ref{update:black2}). At the end we reset all nodes to white.
\begin{algorithm2e}
\begin{footnotesize}
\LinesNumbered
\SetKwInOut{Input}{Input}\SetKwInOut{Output}{Output}
\Input{Graph $G=(V,E)$, vector $d$ of distances from $s$, vector $\sigma$ of number of shortest paths from $s$, batch $\beta$}
\Output{New values of $d(v)$ and $\sigma(v)$ $\forall v \in V$}
\textbf{Assumption:} $color(w)=white\quad \forall w \in V$\;
$Q[]\leftarrow$ array of empty queues\;
\ForEach{$e=\{u,v\} \in \beta, d(u)<d(v)$}
{\label{update:batch1}
	$k \leftarrow d(v)+1$; enqueue $v \rightarrow Q[k]$\;
}\label{update:batch2}
$k \leftarrow 1$\; \label{update:queues1}
\While{there are nodes in $Q[j],  j\geq k$}
{\label{lst2:line:start2}
	\While{$Q[k]\neq \emptyset$}
	{
		dequeue $w\leftarrow Q[k]$\;
		\textbf{if} $color(w)=black$ \textbf{then}			continue\; \label{update:black2}
	
		$con(w) \leftarrow \min_{z : (z,w)\in E} d(z) + 1$\;
		\If{$con(w) = k$}
		{\label{update:if}
			update $d'$ and $d''$\;
			\textbf{if} $d(w) = \infty$  \textbf{then}   $vis(w) \leftarrow vis(w)+1$\; \label{update:vis}
			$d(w) \leftarrow k$; $\sigma(w)\leftarrow 0$; $color(w) \leftarrow black$\; \label{update:black1}
			\ForEach{incident edge $(z,w)$}
			{
				\If{$d(w)=d(z)+1$}
				{\label{update:sp}
					$\sigma(w)\leftarrow \sigma(w)+\sigma(z)$\;
				}\label{update:sp2}
				\If{$d(z)>k$}
				{\label{update:neigh1}
					enqueue $z \rightarrow Q[k+1]$\;
				}\label{update:neigh2}
			}
		}\label{update:if2}
		\Else
		{\label{update:else}
			\If{$d(w) \neq \infty$}
			{
				$d(w)\leftarrow \infty$\;
				$vis(w) \leftarrow vis(w)-1$\; \label{update:vis3}
				\If{vis(w)=0}{
				insert $w$ into $U$\;
				}
				\If{$con(w) \neq \infty$}
				{
					enqueue $w \rightarrow Q[con(w)]$\;
					\ForEach{incident edge $(z,w)$}
					{
						\If{$d(z)>k$}
						{
							enqueue $z \rightarrow Q[k+1]$\;
						}
					}
				}
			}
		}\label{update:else2}
	}
	$k\leftarrow k+1$\;
}\label{update:queues2}
Set to white all the nodes that have been in $Q$\;
\end{footnotesize}
\caption{SSSP update for unweighted graphs (\textsf{updateSSSP-U})}
\label{algo:unweighted}
\end{algorithm2e}
 The replacement of the priority queue with the list of queues decreases the complexity of the SSSP update algorithms for unweighted graphs, which we call \upvdu and \ssspu, in analogy with the ones for weighted graphs.
\begin{lemma}
\label{thm:complexity2}
The time required by \upvdu (\ssspu) to update the distances and \vda (the number of shortest paths) is $O(|\beta |+ ||A ||+d_{\max})$, where $d_{\max}$ is the maximum distance from $s$ reached during the update.
\end{lemma}
\begin{proof}
The complexity of the initialization (Lines~\ref{update:batch1}~-~\ref{update:batch2}) of Algorithm~\ref{algo:unweighted}
is $O(|\beta|)$, as we have to scan the batch. In the main loop (Lines~\ref{update:queues1}~-~\ref{update:queues2}), we scan all the list
of queues, whose final size is $d_{\max}$. 
Every time we extract a node $w$ whose color is not black, we scan all the 
incident edges, therefore this operation is linear in the number of neighbors of $w$. 
If the first time we extract $w$ (say at level $k$) $con(w) $ is equal to $k$, then $w$ will be set to black and will not be scanned anymore.
If the first time we extract $w$, $con(w)$ is instead greater than $k$, $w$ will be inserted into the queue at level $con(w)$ (if $con(w)<\infty$). 
Also, other inconsistent neighbors of $w$ might insert $w$ in one of the queues. However, after the first time $w$ is extracted,
its distance is set to $\infty$, so its neighbors will not be scanned unless $con(w)=k$, in which case 
they will be scanned again, but for the last time, since $w$ will be set to black. To summarize, each affected node and its neighbors can 
be scanned at most twice. The complexity of the
algorithm is therefore $O(|\beta|+\|A\|+d_{\max})$. \qed
\end{proof}

\subsection{Fully-dynamic \vd approximation.}
\label{sub:dyn_vd_approx}
The algorithm keeps track of a \vd approximation for an undirected graph $G$, \ie for each connected component of $G$. It is composed of two phases. In the initialization, we compute an SSSP from a source node $s_i$ for each connected component $C_i$. During the SSSP search from $s_i$, we also compute a \vd approximation $\tilde{\mathit{VD}_i}$ for $C_i$, as described in Sections~\ref{sec:rk} and~\ref{sub:weighted}. In the update, we recompute the SSSPs and the \vd approximations with \upvd (or \upvdu). Since components might split or merge, we might need to compute new approximations, in addition to update the old ones. To do this, for each node, we keep track of the number $vis(v)$ of times it has been visited. This way we discard source nodes that have already been visited and compute a new approximation for components that have become unvisited. 
Algorithm~\ref{algo:vd_approx} describes the initialization. Initially, we put all the nodes in a queue and compute an SSSP from the nodes we extract (Line~\ref{line:sssp}). During the SSSP search, we increase the number of visits $vis(v)$ for all the nodes we traverse (Line~\ref{line:mark_visited}). When extracting the nodes, we skip those that have already been visited (Line~\ref{line:skip}): this avoids computing multiple approximations for the same component. 

In the update (Algorithm~\ref{algo:dyn_vd_approx}), we recompute the SSSPs and the \vd approximations with \upvd (or \upvdu) (Line~\ref{line:upvd}). Since components might split, we might need to add \vd approximations for some new subcomponents, in addition to recomputing the old ones. Also, if components merge, we can discard the superfluous approximations. To do this, we update $vis(v)$ for each node within \upvd (or \upvdu). Before the update, all the nodes are visited exactly once. While updating an SSSP from $s_i$, we increase (decrease) by one $vis(v)$ of the nodes $v$ that become reachable (unreachable) from $s_i$. This way we can skip the update of the SSSPs from nodes that have already been visited (Line~\ref{line:skip}). After the update, for all nodes $v$ that have become unvisited ($vis(v)=0$), we compute a new \vd approximation from scratch (Lines~\ref{line:newvd1} -~\ref{line:newvd2}).
The complexity of the update of the \vd approximation derives from the \vda update in the single components, using \upvd and \upvdu.
\begin{theorem}
\label{thm:complexity_vd}
The time required to update the \vd approximation is $O(n_c\cdot |\beta |\log |\beta |+\sum_{i=1}^{n_c}||A^{(i)}|| \log ||A^{(i)}||)$ in weighted graphs and $O(n_c\cdot |\beta |+\sum_{i=1}^{n_c}||A^{(i)}||+d^{(i)}_{max})$ in unweighted graphs, where $n_c$ is the number of components in $G$ before the update and $A^{(i)}$ is the sum of affected nodes in $C_i$ and their incident edges.
\end{theorem}
\begin{proof}
In the first part (Lines~\ref{vdup:loop1}~-~\ref{vdup:loop2} of Algorithm~\ref{algo:dyn_vd_approx}), we update an SSSP with \upvd or \upvdu for each source node $s_i$ such that $vis(s_i)$ is not greater than 1. Therefore the complexity of the first part is $O(n_c\cdot |\beta |\log |\beta |+\sum_{i=1}^{n_c}||A^{(i)}|| \log ||A^{(i)}||)$ in weighted graphs and $O(n_c\cdot |\beta |+\sum_{i=1}^{n_c}||A^{(i)}||+d^{(i)}_{max})$ in unweighted, by Lemmas~\ref{thm:complexity} and~\ref{thm:complexity2}. Only some of the affected nodes (those whose distance from a source node becomes $\infty$) are inserted into the queue $U$. Therefore the cost of scanning $U$ in Lines~\ref{line:newvd1} -~\ref{line:newvd2} is $O(\sum_{i=1}^{n_c}||A^{(i)}||)$. New SSSP searches are computed for new components that are not covered by the existing source nodes anymore. However, also such searches involve only the affected nodes and each affected node (and its incident edges) is scanned at most once during the search. Therefore, the total cost is $O(n_c\cdot |\beta |\log |\beta |+\sum_{i=1}^{n_c}||A^{(i)}|| \log ||A^{(i)}||)$ for weighted graphs and $O(n_c\cdot |\beta |+\sum_{i=1}^{n_c}||A^{(i)}||+d^{(i)}_{max})$ for unweighted graphs. \qed
\end{proof}

\begin{lemma}
\label{thm:vd_correctness}
At the end of Algorithm~\ref{algo:vd_approx}, $vis(v)=1$ for all $v \in V$, and exactly one \vd approximation is computed for each connected component of $G$.
\end{lemma}
\begin{proof}
Let $v$ be any node. Then $v$ must be scanned by \emph{at least} one source node $s_i$ in the while loop (Lines~\ref{initvd:newsamples1}~-~\ref{initvd:newsamples2}): In fact, either $v$ is visited by some $s_i$ before $v$ is extracted from $U$, or $vis(v)=0$ at the moment of the extraction and $v$ becomes a source node itself. This implies that $vis(v) \geq 1, \ \forall v \in V$.
On the other hand, $vis(v)$ cannot be greater than $1$. In fact, let us assume by contradiction that $vis(v)>1$. This means that there are at least two source nodes $s_i$ and $s_j$ ($i<j$, w.l.o.g.) that are in the same connected component as $v$. Then also $s_i$ and $s_j$ are in the same connected component and $s_j$ is visited during the SSSP search from $s_i$. Then $vis(s_j)=1$ before $s_j$ is extracted from $U$ and $s_j$ cannot be a source node.
Therefore, $vis(v)$ is exactly equal to 1 for each $v \in V$, which means that exactly one \vd approximation is computed for each connected component of $G$.
 \qed
\end{proof}

\begin{proposition}
\label{thm:dyn_vd_correctness}
Let $\mathcal{C}' = \{C'_1,...,C'_{n'_c}\}$ be the set of connected components of $G$ after the update. Algorithm~\ref{algo:dyn_vd_approx} updates or computes exactly one \vd approximation for each $C'_i \in \mathcal{C}'$.
\end{proposition}
\begin{proof}
Let $\mathcal{C} = \{C_1,...,C_{n_c}\}$ be the set of connected components before the update. Let us consider three basic cases (then it is straightforward to see that the proof holds also for combinations of these cases): (i) $C_i \in \mathcal{C}$ is also a component of $\mathcal{C}'$, (ii) $C_i \in \mathcal{C}$ and $C_j \in \mathcal{C}$ merge into one component $C'_k$ of $\mathcal{C}'$, (iii) $C_i \in \mathcal{C}$ splits into two components $C'_j$ and $C'_k$ of $\mathcal{C}'$. In case (i), the \vd approximation of $C_i$ is updated exactly once in the for loop (Lines~\ref{vdup:loop1}~-~\ref{vdup:loop2}). In case (ii), (assuming $i<j$, w.l.o.g.) the \vd approximation of $C'_k$ is updated in the for loop from the source node $s_i \in C_i$. In its SSSP search, $s_i$ visits also $s_j \in C_j$, increasing $vis(s_j)$. Therefore, $s_j$ is skipped and exactly one \vd approximation is computed for $C'_k$. In case (iii), the source node $s_i \in C_i$ belongs to one of the components (say $C'_j$) after the update. During the for loop, the \vd approximation is computed for $C'_j$ via $s_i$. Also, for all the nodes $v$ in $C'_k$, $vis(v)$ is set to 0 and $v$ is inserted into $U$. Then some source node $s'_k \in C'_k$ must be extracted from $U$ in Line~\ref{vdup:extract} and a \vd approximation is computed for $C'_k$. Since all the nodes in $C'_k$ are set to visited during the search, no other \vd approximations are computed for $C'_k$.
 \qed
\end{proof}
\begin{algorithm2e}[h]
 \begin{footnotesize}
\LinesNumbered
\SetKwFunction{initDynamicSSSP}{initApprVD}
\SetKwInOut{Input}{Input}\SetKwInOut{Output}{Output}
\SetKwFunction{SSSP}{SSSP}
\Input{Graph $G=(V,E)$}
\Output{Upper bound on \vd}
$U \leftarrow []$\;
\ForEach{node $v \in V$}
{
		$vis(v)\leftarrow 0$; insert $v$ into $U$\; \label{initvd:queue}
}
$i \leftarrow 1$\;
\While{$U\neq \emptyset$}
{\label{initvd:newsamples1}
	extract $s$ from $U$\;
	\If{$vis(s) = 0$} 
	{\label{line:skip}
		 $s_i \leftarrow s$\;
		$\tilde{\textsf{VD}}_i$ $\leftarrow$ \initDynamicSSSP{$G,s_i$}\;
		$i \leftarrow i+1$\;
	}
}\label{initvd:newsamples2}
$n_C \leftarrow i-1$\;
\vda $\leftarrow \max_{i=1,...,n_C} \tilde{\mathit{VD}}_i$\;
\Return{\vda}\;
 \tcp{\initDynamicSSSP computes \vda and adds 1 to $vis(v)$ of the nodes it visits}
 
\textbf{Function}
\initDynamicSSSP{$G, s$}\\ \Indp {
  	 \SSSP{$G, s$}\; \label{line:sssp}
	$d' \leftarrow \max \{d(s, u) | u \in V, d(u,v) \neq \infty\}$\;
	$d'' \leftarrow \max \{d(s, v) | v \in V, v \neq u, d(s,v) \neq \infty\}$\;
	$\underline{\omega} \leftarrow \min\{\omega(x,y) | (x,y) \in E\}$\;
	\vda $\leftarrow 1 +\frac{d'+d''}{\underline{\omega}}$\;
  	\ForEach{node $w \in V$ s.t. $d(s,w) \neq \infty$} {
		$vis(w) \leftarrow vis(w) + 1$\; \label{line:mark_visited}
	}
  	\Return \vda;
}
	
\end{footnotesize}
\caption{Dynamic \vd approximation (initialization)}
\label{algo:vd_approx}
\end{algorithm2e}

\begin{algorithm2e}[h]
 \begin{footnotesize}
\LinesNumbered
\SetKwFunction{dynamicSSSP}{updateApprVD}
\SetKwFunction{initDynamicSSSP}{initApprVD}
\SetKwInOut{Input}{Input}\SetKwInOut{Output}{Output}
\Input{Graph $G=(V,E)$, vector $vis$}
\Output{New \vd approximation}
$U \leftarrow []$\; 
\ForEach{$s_i$}
{\label{vdup:loop1}
	\If{vis($s_i )> 1$}
	{	
		remove $s_i$ and $\tilde{\mathit{VD}}_i$; decrease $n_C$\; \label{line:skip2}
	}
	\Else{
		 \tcp{\dynamicSSSP updates $vis$, inserts all $v$ for which $vis(v)=0$ into $U$ and recomputes a \vd approximation $\tilde{\mathit{VD}}_i$}
		 $\tilde{\textsf{VD}}_i$ $\leftarrow$ \dynamicSSSP{$G,s_i$} \; \label{line:upvd}
	}
}\label{vdup:loop2}
$i \leftarrow n_C$ \;
\While{$U\neq \emptyset$}
{\label{line:newvd1}
	extract $s'$ from $U$\; \label{vdup:extract}
	\If{$vis(s') = 0$}
	{
		 $s'_i \leftarrow s'$\;
		$\tilde{\textsf{VD}}_i$ $\leftarrow$ \initDynamicSSSP{$G,s'_i$}\;
		$i \leftarrow i+1$; $n_C \leftarrow n_C+1$\;
	}
}\label{line:newvd2}
reset $vis(v)$ to 1 for nodes $v$ such that $vis(v)>1$\;
\vda $\leftarrow \max_{i=1,...,n_C} \tilde{\mathit{VD}}_i$\;
\Return{\vda}
\end{footnotesize}
\caption{Dynamic \vd approximation (\textsf{updateApprVD})}
\label{algo:dyn_vd_approx}
\end{algorithm2e}
\subsection{Combined dynamic BC approximation.}
\label{sub:combined}
Let $G$ be an undirected graph with $n_c$ connected components. In Section~\ref{sub:fully-dyn}, we described an algorithm to update the betweenness approximations in fully-dynamic graphs. If the graph is undirected, we can use the fully dynamic \vd approximation to recompute \vda after a batch, instead of recomputing it from scratch. Then, we could update the $r$ sampled paths with \textsf{updateSSSP} and, if \vda (and therefore $r$) increases, we could sample new paths. 
However, since \textsf{updateSSSP} and \textsf{updateApprVD} share most of the operations, we can ``merge'' them and update at the same time the shortest paths from a source node $s$ and the \vd approximation for the component of $s$. We call this hybrid function \textsf{updateSSSPVD}. Instead of storing and updating $n_c$ SSSPs for the \vd approximation and $r$ SSSPs for the BC scores, we 
recompute a \vd approximation for each of the $r$ samples while recomputing the shortest paths with \textsf{updateSSSPVD}. This way we do not need to compute an additional SSSP for the components covered by the $r$ sampled paths (\ie the components in which the paths lie), saving time and memory. Only for components that are not covered by any of them (if they exist), we compute and store a separate \vd approximation. We refer to such components as $R'$ (and to $|R'|$ as $r'$).

In the initialization (Algorithm~\ref{algo:init}), we first compute the $r$ SSSP, like in \rk (Lines~\ref{init:sampling1}~-~\ref{init:sampling2}). However, we also check which nodes have been visited, as in Algorithm~\ref{algo:vd_approx}. While we compute the $r$ SSSPs, in addition to the distances and number of shortest paths, we also compute a \vd approximation for each of the $r$ source nodes and increase $vis(v)$ of all the nodes we visit during the sources with \textsf{initApprVD} (Line~\ref{init:dynsssp}). Since it is possible that the $r$ shortest paths do not cover all the components of $G$, we compute an additional VD approximation for nodes in the unvisited components, like in Algorithm~\ref{algo:vd_approx} (Lines~\ref{init:newsamples1}~-~\ref{init:newsamples2}). Basically we can divide the SSSPs into two sets: the set $R$ of SSSPs used to compute the $r$ shortest paths and the set $R'$ of SSSPs used for a \vd approximation in the components that were not scanned by the initial $R$ SSSPs. We call $r'$ the number of the SSSPs in $R'$.

The BC update after a batch is described in Algorithm~\ref{algo:bcupdate}. First (Lines~\ref{up:update1} -~\ref{up:update2}), we recompute the shortest paths like in Algorithm~\ref{algo:main_algo}: we update the SSSPs from each source node $s$ in $R$ and we replace the old shortest path with a new one (subtracting $1/r$ from the nodes in the old shortest path and adding $1/r$ to those in the new shortest path).
To update the SSSPs, we use the fully-dynamic \textsf{updateSSSPVD} that updates also the \vd approximation and keeps track of the nodes that become unvisited.
Then (Lines~\ref{up:ext1} -~\ref{up:ext2}), we add a new SSSP to $R'$ for each component that has become unvisited (by both $R$ and $R'$). After this, we have at least a \vd approximation for each component of $G$. We take the maximum over all these approximations and recompute the number of samples $r$ (Lines~\ref{up:recompute1} -~\ref{up:recompute2}). If $r$ has increased, we need to sample new paths and therefore new SSSPs to add to $R$. Finally, we normalize the BC scores, \ie we multiply them by the old value of $r$ divided by the new value of $r$ (Line~\ref{up:norm}).
We refer to the algorithm for unweighted graphs as \da and the one for weighted as \daw. The difference between \da and \daw is the way the SSSPs and the \vd approximation are updated: in \da we use \upvdu and in \daw \upvd. \begin{theorem}
\label{thm:correctness_bc}
Algorithm~\ref{algo:bcupdate} preserves the guarantee on the maximum absolute error, \ie naming $c'_{B}(v)$ and $\tilde{c}'_B(v)$ the new exact and approximated BC values, respectively, $\Pr(\exists v\in V\: s.t.\:|c'_{B}(v)-\tilde{c}'_B(v)|>\epsilon)<\delta$.
\end{theorem}
\begin{proof}
Let $G$ be the old graph and $G'$ the modified graph after the batch of edge updates. Let $p'_{xy}$ be a shortest path of $G'$ between nodes $x$ and
$y$. To prove the theoretical guarantee, we need to prove that the probability of any sampled path $p'_{(i)}$ is equal to $p'_{xy}$ (\ie that the algorithms adds $1/r'$ to the nodes in $p'_{xy}$) is $\frac{1}{n(n-1)}\frac{1}{\sigma'_{x}(y)}$.
Algorithm~\ref{algo:bcupdate} replaces the first $r$ shortest paths with other shortest paths $p'_{(1)},...,p'_{(r)}$ between the same node pairs (Lines~\ref{up:new1}~-~\ref{up:new2}) using Algorithm~\ref{algo:main_algo}, for which we already proved that $\Pr(p'_{(k)}=p'_{xy})=\frac{1}{n(n-1)}\frac{1}{\sigma'_{x}(y)}$ (Lemma~\ref{lem:2}). The additional $\Delta r$ shortest paths (Line~\ref{up:new_paths}) are recomputed from scratch with \rk, therefore also in this case $\Pr(p'_{(k)}=p'_{xy})=\frac{1}{n(n-1)}\frac{1}{\sigma'_{x}(y)}$ by Lemma 7 of~\cite{DBLP:conf/wsdm/RiondatoK14}.
 \qed
\end{proof}
\begin{algorithm2e}
 \begin{footnotesize}
\LinesNumbered
\SetKwData{B}{$\tilde{c}_B$}\SetKwData{VD}{VD}
\SetKwFunction{getVertexDiameter}{getApproxVertexDiameter}
\SetKwFunction{sampleUniformNodePair}{sampleUniformNodePair}
\SetKwFunction{computeExtendedSSSP}{computeExtendedSSSP}
\SetKwFunction{initDynamicSSSP}{initApprVD}
\SetKwFunction{initSSSPVD}{initSSSPVD}
\SetKwInOut{Input}{Input}\SetKwInOut{Output}{Output}
\Input{Graph $G=(V,E)$, source node $s$, number of iterations $r$, batch $\beta$}
\Output{Approximated BC values $\forall v\in V$}
\ForEach{node $v \in V$}
{
	\B$(v)\leftarrow 0$; $vis(v) \leftarrow 0$\; \label{init:init}
}
\vda$\leftarrow$\getVertexDiameter{$G$}\; \label{init:sampling1}
$r \leftarrow (c/\epsilon^2) (\lfloor \log_2(\tilde{\mathit{VD}}-2)\rfloor +\ln(1/\delta))$\;

\For{$i \leftarrow 1$ \KwTo $r$}{
	\label{sampling1}
	$(s_i,t_i)\leftarrow$ \sampleUniformNodePair{$V$}\;
	$\tilde{\mathit{VD}}_i \leftarrow$ \initDynamicSSSP{$G,s_i,$}\; \label{init:dynsssp}
	$v \leftarrow t_i$\;\label{paths1}
	$p_{(i)}\leftarrow$ empty list\;
	$P_{s_i}(v) \leftarrow\{ z : \{z,v\}\in E \cap d_{s_i}(v) = d_{s_i}(z)+\omega(\{z,v\}) \}$\;
	\While{$P_{s_i}(v) \neq \{s_i\}$}
	{
		\mbox{sample $z \in P_{s_i}(v)$ with probability $\sigma_{s_i}(z)/\sigma_{s_i}(v)$}\;
		$\B(z)\leftarrow \B(z)+1/r$\;
		add $z\rightarrow p_{(i)}$;
		$v\leftarrow z$\;
		$P_{s_i}(v) \leftarrow\{ z : \{z,v\}\in E \cap d_{s_i}(v) = d_{s_i}(z)+\omega(\{z,v\}) \}$\;
	}\label{paths2}
}\label{sampling2} \label{init:sampling2}
$U \leftarrow V$\;
$i \leftarrow r+1$\;
\While{$U\neq \emptyset$}
{\label{init:newsamples1}
	extract $s'$ from $U$\;
	\If{$vis(s') = 0$}
	{
		 $s'_i \leftarrow s'$\;
		$\tilde{\mathit{VD}}_i \leftarrow$ \initDynamicSSSP{$G,s'_i$}\;
		$i \leftarrow i+1$\;
	}
}\label{init:newsamples2}
$r' \leftarrow r-i-1$\;
\Return{$\{(v,\B(v)):\: v\in V\}$}
\end{footnotesize}
\caption{BC initialization}
\label{algo:init}
\end{algorithm2e}

\begin{algorithm2e}
\begin{footnotesize}
\LinesNumbered
\SetKwData{B}{$\tilde{c}_B$}
\SetKwFunction{dynamicSSSP}{updateApprVD}
\SetKwFunction{dynamicSSSPVD}{updateSSSPVD}
\SetKwFunction{initDynamicSSSP}{initApprVD}
\SetKwInOut{Input}{Input}\SetKwInOut{Output}{Output}
\Input{Graph $G=(V,E)$, source node $s$, number of iterations $r$, batch $\beta$}
\Output{New approximated BC values $\forall v\in V$}
$U \leftarrow []$\;
\For{$i \leftarrow 1$ \KwTo $r$}
{\label{up:update1}
	$d^{old}_i\leftarrow d_{s_i}(t_i)$\;
	$\sigma^{old}_i\leftarrow \sigma_{s_i}(t_i)$\;
	 \tcp{\dynamicSSSPVD updates $vis$, inserts all $v :\: vis(v)=0$ into $U$ and updates the \vd approximation}
	$\tilde{\mathit{VD}}_i \leftarrow$ \dynamicSSSPVD{$G, s_i, \beta$}\; \label{up:dynsssp}
	%\tcp{If the shortest paths between $s_i$ and $t_i$ have changed, we sample a new one}
	\tcp { we replace the shortest path between $s_i$ and $t_i$}
			%\tcp{We subtract $1/r$ from all nodes in the old shortest path and we add $1/r$ to all nodes in the new shortest path}
			\ForEach{$w \in p_{(i)}$}
			{
				\B($w$) $\leftarrow \B(w)-1/r$\; \label{up:sub}
			}
			$v \leftarrow t_i$\;
			$p_{(i)}\leftarrow$ empty list\;
			$P_{s_i}(v) \leftarrow\{ z : \{z,v\}\in E \cap d_{s_i}(v) = d_{s_i}(z)+\omega(\{z,v\}) \}$\;\label{up:pred1}
			\While{$P_{s_i}(v) \neq \{s_i\}$}
			{\label{up:new1}
				\mbox{sample $z \in P_{s_i}(v)$ with probability $=\sigma_{s_i}(z)/\sigma_{s_i}(v)$}\;
				$\B(z)\leftarrow \B(z)+1/r$\;
				add $z$ to $p_{(i)}$\;
				$v\leftarrow z$\;
				$P_{s_i}(v) \leftarrow\{ z : \{z,v\}\in E \cap d_{s_i}(v) = d_{s_i}(z)+\omega(\{z,v\}) \}$\;\label{up:pred2}
			}\label{up:new2}
}\label{up:update2}
\For{$i \leftarrow r+1$ \KwTo $r+r'$}
{\label{up:update3}
		$\tilde{\mathit{VD}}_i \leftarrow$ \dynamicSSSP{$G, s_i, \beta$}\;
}\label{up:update4}
$i \leftarrow r+r'+1$\;\label{ext1}
\While{$U\neq \emptyset$}
{\label{up:ext1}
	extract $s'$ from $U$\;
	\If{$vis(s') = 0$}
	{
		 $s'_i \leftarrow s'$\;
		$\tilde{\mathit{VD}}_i \leftarrow$ \initDynamicSSSP{$G, s'_i$}\;
		$i \leftarrow i+1$; $r' \leftarrow r'+1$\;
	}
}\label{up:ext2}
\tcp{compute the maximum over all the ${\mathit{VD}}_i$ computed by \dynamicSSSP}
\vda $\leftarrow \max_{i=1,...,r+r'} \tilde{\mathit{VD}}_i$\; \label{up:recompute1}
$r_{\text{new}} \leftarrow (c/\epsilon^2) (\lfloor \log_2(\tilde{\mathit{VD}}-2)\rfloor +\ln(1/\delta))$\; \label{up:recompute2}
\If{$r_{\text{new}}>r$}
{
	sample new paths\; \label{up:new_paths}
	\ForEach{$v \in V$}
	{
		$\B(v) \leftarrow \B(v) \cdot r/r_{\text{new}}$ \label{up:norm}
	}
$r \leftarrow r_{\text{new}}$\;
}\label{ext4}
\Return{$\{(v,\B(v)):\: v\in V\}$}
\end{footnotesize}
\caption{Dynamic update of BC approximation (\textsf{DA})}
\label{algo:bcupdate}
\end{algorithm2e}

\subsection{Complexity of the dynamic BC algorithms.}
In this section we presented different algorithms for updating BC approximations after batches of edge updates. Algorithm~\ref{algo:main_algo} can be used on graphs for which we can be sure that the vertex diameter cannot increase after a batch of edge updates. This includes, for example, unweighted connected graphs on which only edge insertions are allowed. We refer to the unweighted version of this algorithm as \ia (incremental approximation) and to the weighted version as \iaw (incremental approximation weighted).
On general directed graphs, we can use the algorithms described in Section~\ref{sub:fully-dyn}. We name the unweighted version \dad (dynamic approximation directed) and the weighted one \dadw.
Finally, for undirected graphs, we can use the optimized algorithms presented in Section~\ref{sub:combined} (Algorithm~\ref{algo:bcupdate}), to which we refer as \da and \daw. Theorem~\ref{theo:complexity} presents the complexities of all the BC update algorithms. In the following, we name $||A^{(i)}||$ the sum of affected nodes and their incident edges in the $i$th sampled SSSP. We also name $r$ the number of samples. In case we need to sample new additional paths after the update (in \dad, \dadw, \da and \daw), we refer to the difference between the value of $r$ before and after the batch as $\Delta r$. In \da and \daw, we call $r'$ the number of additional samples necessary for the \vd approximation.
\begin{theorem}
\label{theo:complexity}
Given a graph $G=(V,E)$ with $n$ nodes and $m$ edges, the time required by the different algorithms to update the BC approximations after a batch $\beta$ are the following:
\begin{enumerate}[(i)]
\item \ia: $O(r\cdot |\beta | + \sum_{i=1}^{r} (||A^{(i)}||+d_{\max}^{(i)})$
\item \iaw: $O(r \cdot |\beta |\log |\beta | + \sum_{i=1}^{r} ||A^{(i)}|| \log ||A^{(i)}||)$ 
\item \dad: $O(r \cdot |\beta | + \sum_{i=1}^{r} (||A^{(i)}||+d_{\max}^{(i)}) + (\Delta r+1)(n+m))$
\item \dadw: $O((r \cdot |\beta |\log |\beta | + \sum_{i=1}^{r} ||A^{(i)}|| \log ||A^{(i)}|| + (\Delta r+1) (n \log n+m))$
\item \da: $O((r+r') |\beta | + \sum_{i=1}^{r+r'} (||A^{(i)}||+d_{\max}^{(i)}) + \Delta r(n+m))$
\item \daw: $O((r+r') |\beta |\log |\beta | + \sum_{i=1}^{r+r'} ||A^{(i)}|| \log ||A^{(i)}|| + \Delta r(n\log n+m))$
\end{enumerate}
\end{theorem}
\begin{proof}
We prove each case separately.
\begin{enumerate}[(i)]
\item \ia updates each sampled path with \textsf{updateSSSP-U}. Therefore, the total complexity is the sum of the times required to update each of the $r$ paths, \ie $O(r\cdot |\beta | + \sum_{i=1}^{r} (||A^{(i)}||+d_{\max}^{(i)})$.
\item Same as (i), with the only difference that we use \textsf{updateSSSP-W} for weighted graphs.
\item In \dad, we need to update the existing $r$ samples, exactly as in \ia. In addition to that, we might need to sample new $\Delta r$ additional paths using a BFS, whose complexity is $O(n+m)$. Also, we need to recompute the upper bound on \vd, whose complexity is also $O(n+m)$ (see Section~\ref{sub:directed}). Therefore, in this case we have to add an additional $O((\Delta r +1)(n+m))$ term to the complexity of \ia.
\item Similarly to (iii), we need to sample $\Delta r$ additional paths, but in weighted graphs the cost of a SSSP is $O(n \log n+m)$. Also the \vd approximation described in Section~\ref{sub:dir_weighted} requires $O(n \log n+m)$ time.
\item Let $\Delta r'$ be the difference between the values of $r'$ before and after the batch. After processing $\beta$, we might need to sample new paths for the betweenness approximation ($\Delta r>0$) and/or sample paths in new components that are not covered by any of the sampled paths ($\Delta r'>0$). Then, the complexity for the betweenness approximation update is $O(r\cdot |\beta | + \sum_{i=1}^r (||A^{(i)}||+d_{\max}^{(i)})) + O(\Delta r (n+m))$. The \vd update requires $O(r'\cdot |\beta | + \sum_{i=1}^{r'} (||A^{(i)}||+d_{\max}^{(i)}))$ to update the \vd approximation in the already covered components and $\sum_{i=1}^{\Delta r}(|V_i|+|E_i|)$ for the new ones, where $V_i$ and $E_i$ are nodes and edges of the $i$th component, respectively. From this derives the total complexity.
\item Same as (v), using \textsf{updateSSSP-W}, \textsf{approxVD-W}.
\end{enumerate}
\qed
\end{proof}
Notice that, if \vda does not increase, $\Delta r = 0$ and the complexities of \da and \dad (\daw and \dadw, respectively) are the same as the only-incremental algorithm \textsf{IA} (\textsf{IAW}, respectively). This case includes, for example, connected graphs subject to a batch of only edge insertions, or any batch that neither splits the graph into more components nor increases \vd. Also, notice that in the worst case the complexity can be as bad as recomputing from scratch. However, no dynamic SSSP (and so probably also no BC approximation) algorithm exists that is faster than recomputation on all graphs.

%%%%%%%%%%%%%%%%%%%%%%%%%%%%%%%%%%
\section{Experiments}
\label{sec:experimental}

%\subsection{Experimental setup}
%\label{sec:exp_setup}
\paragraph{Implementation and settings.} For an experimental comparison, we implemented our six approaches \ia, \iaw, \dad, \dadw, \da, \daw, as well as the static approximation \textsf{RK}~\cite{DBLP:conf/wsdm/RiondatoK14}. In addition to that, we implemented the static exact \textsf{BA} and the dynamic exact algorithms \textsf{GMB}~\cite{DBLP:conf/socialcom/GreenMB12} and \textsf{KDB}~\cite{DBLP:journals/tkde/KourtellisMB15}, which were shown to have the best speedups on unweighted graphs. For a comparison on weighted graphs, we also implemented the algorithm by Nasre 
\etal~\cite{DBLP:conf/mfcs/NasrePR14} (\textsf{NPR}). In the implementation of \rk, we used the optimization proposed in~\cite{DBLP:conf/wsdm/RiondatoK14}, stopping all the SSSP searches once the target node has been found. Also, we computed the number of samples using our new bounds on \vd presented in Section~\ref{sec:new_vd_approx}.
We implemented all algorithms in C++, building on the open-source \textit{NetworKit} framework~\cite{DBLP:journals/corr/StaudtSM14}.
In all experiments we fix $\delta$ to 0.1 while the error bound $\epsilon$ varies.
The machine we employ is used for its 256 GB RAM -- the comparison to exact approaches requires a substantial amount
of memory. Of the machine's 2 x 8 Intel(R) Xeon(R) E5-2680 cores running at 2.7 GHz we use only one; all computations are sequential to make the comparison to previous work more meaningful.
\begin{table*}[t]
\begin{center}
\begin{footnotesize}

  \begin{tabular}{ | l | l | r | r | l |}
    \hline

    Graph 					& Type 	& Nodes 		& Edges 		&  Type\\ \hline
    \texttt{ca-GrQc}		 		& coauthorship 			& 5 242		& 14 496	 	&Unweighted, Undirected	\\
     \texttt{p2p-Gnutella09}			& file sharing			& 8 114		& 26 013		&Unweighted, Directed\\	
    \texttt{ca-HepTh} 				&  coauthorship 		& 9 877		& 25 998		&Unweighted, Undirected	\\    
    \texttt{PGPgiantcompo} 		& social / web of trust 	&10 680		& 24 316		&Unweighted, Undirected	\\ 
    \texttt{as-22july06}		 		& internet 				& 22 963		& 48 436	 	&Unweighted, Undirected	\\
    \hline
  \end{tabular}
  \end{footnotesize}
\end{center}
  \caption{Overview of small real-world networks used in the experiments.}
  \label{table:small}
  % \vspace{-4ex}
\end{table*}
 %\vspace{-2ex}
\begin{table*}[t]
\begin{center}
\begin{footnotesize}
  \begin{tabular}{ | l | l | r | r | r |}
    \hline
    Graph 						& Type 				& Nodes 			& Edges 			&  Type\\ \hline
   
    \texttt{repliesDigg}				& communication	& 30,398			& 85,155 			& Weighted	\\
    \texttt{emailSlashdot}			& communication	& 51,083 			& 116,573			& Weighted	\\ 
    \texttt{emailLinux}				& communication	& 63,399 			& 159,996			& Weighted	\\
    \texttt{facebookPosts}			& communication	& 46,952			& 183,412 		& Weighted	\\
    \texttt{emailEnron}				& communication	& 87,273 			& 297,456			& Weighted	\\
    \texttt{facebookFriends}			& friendship		& 63,731			& 817,035 		& Unweighted	\\
    \texttt{arXivCitations} 			& coauthorship 		& 28,093			& 3,148,447		& Unweighted	\\
    \texttt{englishWikipedia}		& hyperlink		& 1,870,709		& 36,532,531 		& Unweighted	\\

    \hline
  \end{tabular}
  \end{footnotesize}
\end{center}
  \caption{Overview of real dynamic graphs used in the experiments, taken from \url{http://konect.uni-koblenz.de/}.}
  \label{table:graphs}
 %  \vspace{-4ex}
\end{table*}

%\vspace{-8ex}
\paragraph{Data sets.} We use both real-world and synthetic networks. For our experiments on the accuracy and for comparison with the exact algorithms, we use relatively small networks, on which also the non-scalable algorithms can be executed. These networks are summarized in Table~\ref{table:small} and are publicly available from the collection compiled for the 10th DIMACS Challenge~\cite{DBLP:reference/snam/BaderMS0KW14} (\url{http://www.cc.gatech.edu/dimacs10/downloads.shtml}) and from the SNAP collection (\url{http://snap.stanford.edu}). Due to a shortage of actual dynamic networks in this size range, we simulate dynamics by removing a small fraction of random edges and adding them back in batches. We also use synthetic networks obtained with the Dorogovtsev-Mendes generator, a simple model for networks with power-law degree distribution~\cite{dorogovtsev2003evolution}. 

To compare the running times of the scalable algorithms (\rk and our dynamic algorithms), we use real dynamic networks, taken from The Koblenz Network Collection (KONECT)~\cite{DBLP:conf/www/Kunegis13} and summarized in Table~\ref{table:graphs}. 
All the edges of the KONECT graphs are characterized by a time of arrival. In case of multiple edges between two nodes, we extract two versions of the graph: one unweighted, where we ignore additional edges, and one weighted, where we replace the set  $E_{st}$ of edges between two nodes with an edge of weight $1/|E_{st}|$ (more tightly coupled nodes receive a smaller distance). 
In our experiments, we let the batch size vary from 1 to 1024 and for each batch size, we average the running times over 10 runs.
Since the networks do not include edge deletions, we implement additional simulated dynamics. In particular, we consider the following experiments. 
(i) \textit{Real dynamics.} We remove the $x$ edges with the highest timestamp from the network and we insert them back in batches, in the order of timestamps. 
(ii) \textit{Random insertions and deletions.} We remove $x$ edges from the graph, chosen uniformly at random. To create batches of both edge insertions and deletions, we add back the deleted edges with probability $1/2$ and delete other random edges with probability $1/2$. 
(iii) \textit{Random weight changes.} In weighted networks, we choose $x$ edges uniformly at random and we multiply their weight by a random 
value in the interval $(0,2)$.

To study the scalability of the methods, we also use synthetic graphs obtained with a generator based on a unit-disk graph model in hyperbolic geometry~\cite{LoozMP15generating}, where edge insertions and deletions are obtained by moving the nodes in the hyperbolic plane. The networks produced by the model were shown to have many properties of real complex networks, like small diameter and power-law degree distribution (for details and references the interested reader is referred to von Looz \etal~\cite{LoozMP15generating}). We generate seven networks, with $|E|$ ranging from about $2\cdot 10^4$ to about $2 \cdot 10^7$ and $|V|$ approximately equal to $|E|/10$.

We also compare our new upper bound on \vd for directed graphs presented in Section~\ref{sub:directed} with the one used in \rk. For this, we use directed real-world graphs of different sizes taken from the SNAP collection.

\subsection{Accuracy.}
\begin{figure}[htb]
\begin{center}
\includegraphics[width = 0.45\textwidth]{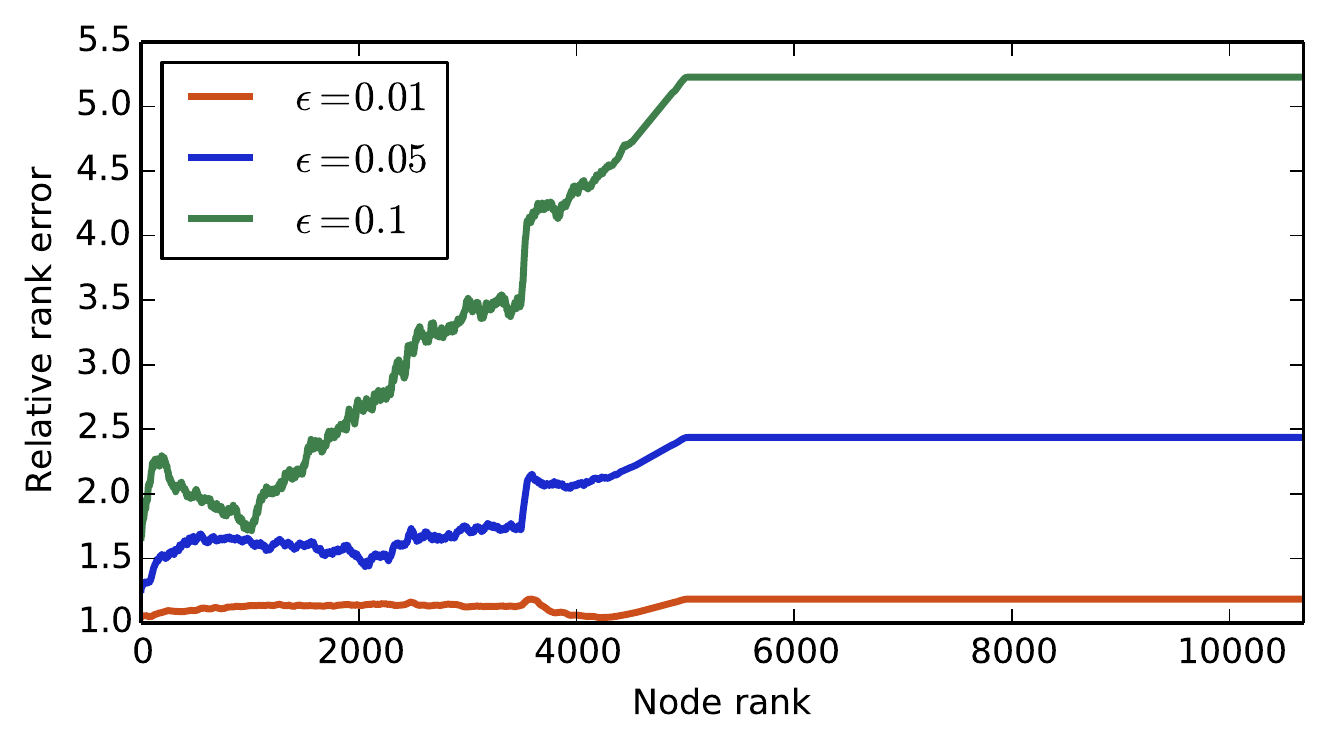}
\includegraphics[width = 0.45\textwidth]{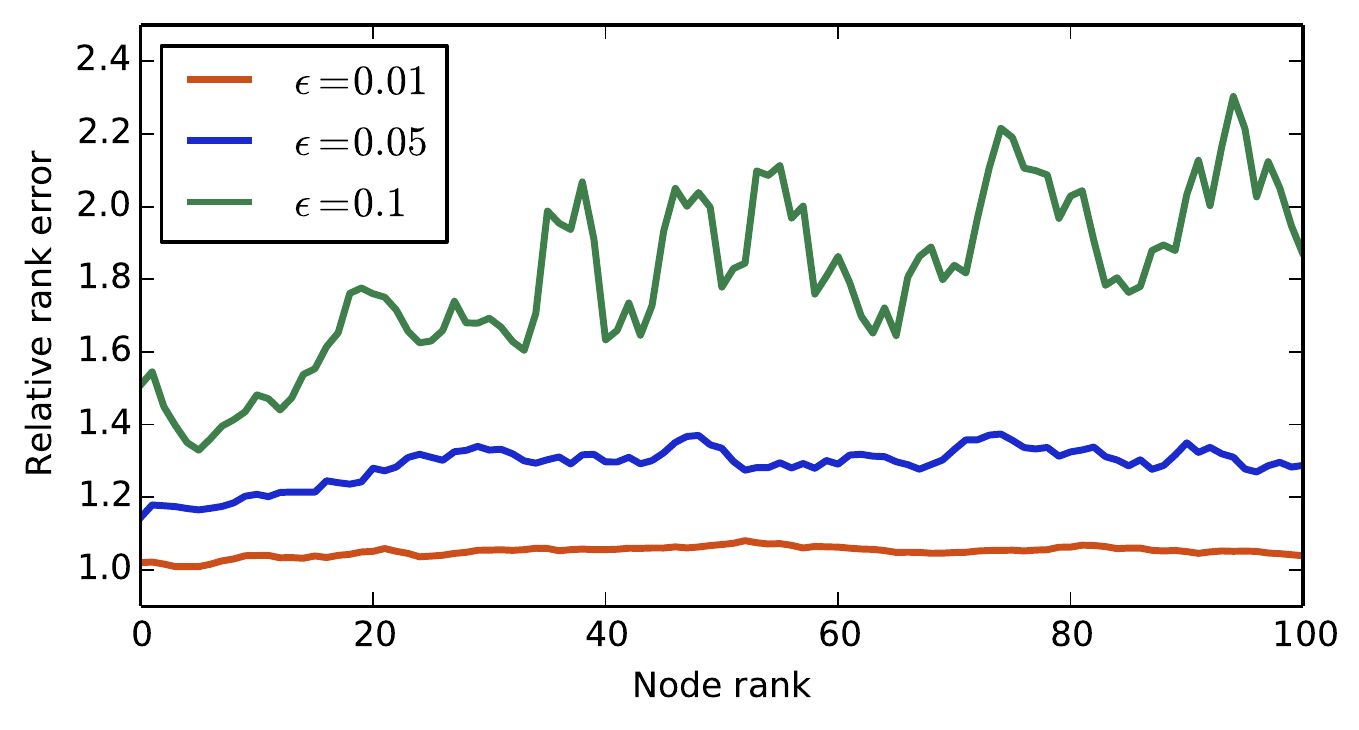}
\caption{Relative rank error on \texttt{PGPgiantcompo} for nodes ordered by rank. Left: relative errors of all nodes. Right: relative errors of the 100 nodes with highest betweenness.}
\label{fig:rank}
\end{center}
\end{figure}
\begin{table*}[h!]
\begin{center}
\begin{footnotesize}
  \begin{tabular}{ | l | r | r | r | r |  r |}
    \hline
    Graph 						&\texttt{{ca-GrQc}}	& \texttt{ca-HepTh} & \texttt{PGPgiantcompo}	&  \texttt{as-22july06} & \texttt{p2p-Gnutella09}  \\ \hline
    max. error ($\epsilon = 0.1$)		& 1.70e-02		& 1.69e-02			& 3.10e-02	& 3.22e-02	& 1.56e-02	\\ 
    max. error ($\epsilon = 0.05$)	& 9.12e-03	 	& 7.62e-03			& 1.38e-02	&  1.60e-02	& 6.55e-03	\\
    max. error ($\epsilon = 0.01$)	& 1.67e-03		& 1.41e-03			&2.99e-03		& 3.45e-03	& 1.23e-03	\\
    avg. error ($\epsilon = 0.1$)		& 4.55e-04		& 3.87e-04			& 4.56e-04	& 8.55e-05	& 5.92e-04	\\
    avg. error ($\epsilon = 0.05$)		& 2.42e-04		& 2.10e-04			& 2.54e-04 	&  5.35e-05	& 3.15e-04	\\
    avg. error ($\epsilon = 0.01$)		& 4.63e-05		& 4.29e-05			& 5.10e-05 	& 1.33e-05	& 6.55e-05	\\
    \hline
  \end{tabular}
 \end{footnotesize}
\end{center}
  \caption{Maximum and average absolute errors on real networks for different values of $\epsilon$ ($\delta=0.1$). The values are averaged over 10 runs.}
  \label{table:errors}
\end{table*}

%\vspace{-4ex}

We consider the accuracy of the approximated centrality scores both in terms of absolute error and, more importantly, the preservation of the ranking order of nodes.
Since we only replace the samples without changing their number, our dynamic algorithm has exactly the same accuracy as \textsf{RK}. 
The authors of~\cite{DBLP:conf/wsdm/RiondatoK14} study the behavior of \textsf{RK} also experimentally, considering the average and maximum estimation error on a small set of real graphs. 
We study the experimental errors on additional graphs. For our tests we use the networks summarized in Table~\ref{table:small} and Dorogovtsev-Mendes graphs of several sizes. Our results confirm those of~\cite{DBLP:conf/wsdm/RiondatoK14} in the sense that the measured absolute errors are \textit{always} below the guaranteed maximum error $\epsilon$ and the measured average error is often orders of magnitude smaller than $\epsilon$. Table~\ref{table:errors} shows the measured errors for the real networks.
We also study the \textit{relative rank error} introduced by Geisberger \etal~\cite{DBLP:conf/alenex/GeisbergerSS08} (\ie $\max \{\rho, 1/\rho\}$, denoting $\rho$ the ratio between the estimated rank and the
true rank), which we consider the most relevant measure of the quality of the approximations. Figure~\ref{fig:rank} shows the results for \texttt{PGPgiantcompo}, a similar trend can be observed on our other test instances as well. On the left, we see the errors for the whole set of nodes (ordered by exact rank) and, on the right, we focus on the top 100 nodes. The straight lines in the plot on the left correspond to nodes with betweenness 0, which are therefore undistinguishable.
The plots show that for a small value of $\epsilon$ (0.01), the ranking is very well preserved over all the positions. With higher values of $\epsilon$, the rank error of the nodes with low betweenness increases, as they are harder to approximate. However, the error remains small for the nodes with highest betweenness, the most important ones for many applications.

\subsection{New upper bound on \vd for directed graphs.}
We compute the new upper bound on \vd for directed graphs presented in Section~\ref{sub:directed} and compare it with the upper bound used in \rk~\cite{DBLP:conf/wsdm/RiondatoK14}, \ie the size of the largest weakly connected component. All the networks used in the experiment are real directed graphs. Since finding \vd exactly would be expensive in most of the graphs we used, we also compute a lower bound on \vd, by sampling nodes in the graph, computing their eccentricity (\ie, the maximum distance reachable from the node), and adding 1 to it. In Table~\ref{table:new_bound}, we report this lower bound (${\mathit{VD}}^{\star}$), our new upper bound ($\tilde{\mathit{VD}}$) and the one used in \rk ($\tilde{\mathit{VD}}_{\text{RK}}$). The results show that $\tilde{\mathit{VD}}$ is always several orders of magnitude smaller than $\tilde{\mathit{VD}}_{\text{RK}}$ and never more than a factor 4 from the lower bound ${\mathit{VD}}^{\star}$ (and therefore also from \vd). This difference is mitigated by the logarithm in the number of samples required for the approximation (see Eq.~(\ref{sample_size})). However, Table~\ref{table:new_bound} shows that $\tilde{\mathit{VD}}$ is almost always more than $2^{10}$ times smaller than $\tilde{\mathit{VD}}_{\text{RK}}$, resulting in at least 10 times less samples required for the theoretical guarantee. 
\begin{table*}[h!]
\begin{center}
\begin{footnotesize}
  \begin{tabular}{ | l | r | r | r | r | r |}
    \hline
    Graph 					& Nodes 		& Edges 	 	& ${\mathit{VD}}^{\star}$	&  $\tilde{\mathit{VD}}$ & $\tilde{\mathit{VD}}_{\text{RK}}$  \\ \hline
    \texttt{p2p-Gnutella24}		&	26 518	& 65 369	 	& 20			& 47		    &  26 498\\    
    \texttt{soc-Epinions1}		&	75 879	& 508 837	 	& 13			& 41		    &  75 877\\
     \texttt{slashdot081106}		&	77 356	& 516 575	 	& 14			& 39		    &  77 349\\
     \texttt{twitter-comb}		&	81 306	& 1 768 149	& 9			& 34		    &  81 306\\
     \texttt{slashdot090216}		&	81 870	& 545 671		& 14			& 40		    &  81 866\\
      \texttt{amazon0302}		&	262 111	& 1 234 877	& 71			& 183	    &  262 111\\
       \texttt{email-EuAll}		&	265 214	& 420 045		& 9			& 23		    &  224 832\\
    \hline
  \end{tabular}
 \end{footnotesize}
\end{center}
  \caption{Lower bound on \vd (${\mathit{VD}}^{\star}$) and upper bounds (our new bound $\tilde{\mathit{VD}}$ and the one proposed in \rk, $\tilde{\mathit{VD}}_{\text{RK}}$) on real-world networks.}
  \label{table:new_bound}
\end{table*}

\subsection{Running times.}
In this section, we discuss the running times of all the algorithms, the speedups of the exact incremental approaches (\textsf{GMB}, \textsf{KDB}) on \textsf{BA} and the speedups of our algorithms on \textsf{RK}. (Note that the term \emph{speedup} is used in this
paper to compare different algorithms, not sequential vs parallel execution.)

In all of our tests on relatively small graphs (Table~\ref{table:small}), \textsf{KDB} performs worse than \textsf{GMB}, therefore we only report the results of \textsf{GMB}.
 Figure~\ref{pgp} shows the behavior of the four algorithms on \texttt{PGPgiantcompo}. Since the graph is connected and we are considering batches of edge insertions, \ia and \da are identical in this case (therefore in Figure~\ref{pgp} we refer to our algorithm as \ia). On the left, the figure shows the running times of the four algorithms, whereas the plot on the right reports the speedups of \textsf{GMB} on \textsf{BA} and of \ia on \rk.
  \textsf{GMB} can only process edges one by one, therefore its running time increases linearly with the batch size, becoming slower than the static algorithm already with a batch size of 64. Our algorithm shows much better speedups and proves to be significantly faster than recomputation even with a batch of size 1024. The reasons for our high speedup are mainly two: First, we process the updates in a batch, processing only once the nodes affected by multiple edge insertions. Secondly, our algorithm does not need to recompute the dependencies, in contrast to all dynamic algorithms based on \textsf{BA} (\ie all existing dynamic exact algorithms). For each SSSP, the dependencies need to be recomputed not only for nodes whose distance or number of shortest paths from the source has changed after the edge insertion(s), but also for all the intermediate nodes in the old shortest paths, even if their distance and number of shortest paths are not affected. This number is significantly higher, since for every node which changes its distance or increases its number of shortest paths, the dependencies of \textit{all} the nodes in \textit{all} the old shortest paths are affected.

Results on the other small graphs and on synthetic Dorogovtsev-Mendes graphs are analogous to those shown in Figure~\ref{pgp}. 
Figure~\ref{small_synth} (left) shows the speedups of \ia on \rk and of \textsf{GMB} on \textsf{BA}, averaged over 10 synthetic graphs. 
We see in Figure~\ref{small_synth} (right) that our algorithm is significantly faster than \textsf{GMB}, and the speedup of \textsf{IA} on \textsf{GMB} clearly increases with the batch size. 
We also compared \textsf{NPR} with \iaw on small weighted graphs. Since \textsf{NPR} performs very poorly on all tested instances, we do not report the results here. However, they can be found in our conference paper~\cite{DBLP:conf/alenex/BergaminiMS15}.
\textsf{GMB} and \textsf{NPR} also have very high memory requirements ($\Theta(n^2+mn)$), which makes the algorithms unusable on networks with more than a few thousand edges. The memory requirement is the same also for all other existing dynamic algorithms, with the exception of \textsf{KDB}, which requires $\Theta(n^2)$ memory, still impractical for large networks. 
%
%\vspace{-6ex}
\begin{figure}[h!]
\begin{center}
\includegraphics[width = 0.45\textwidth]{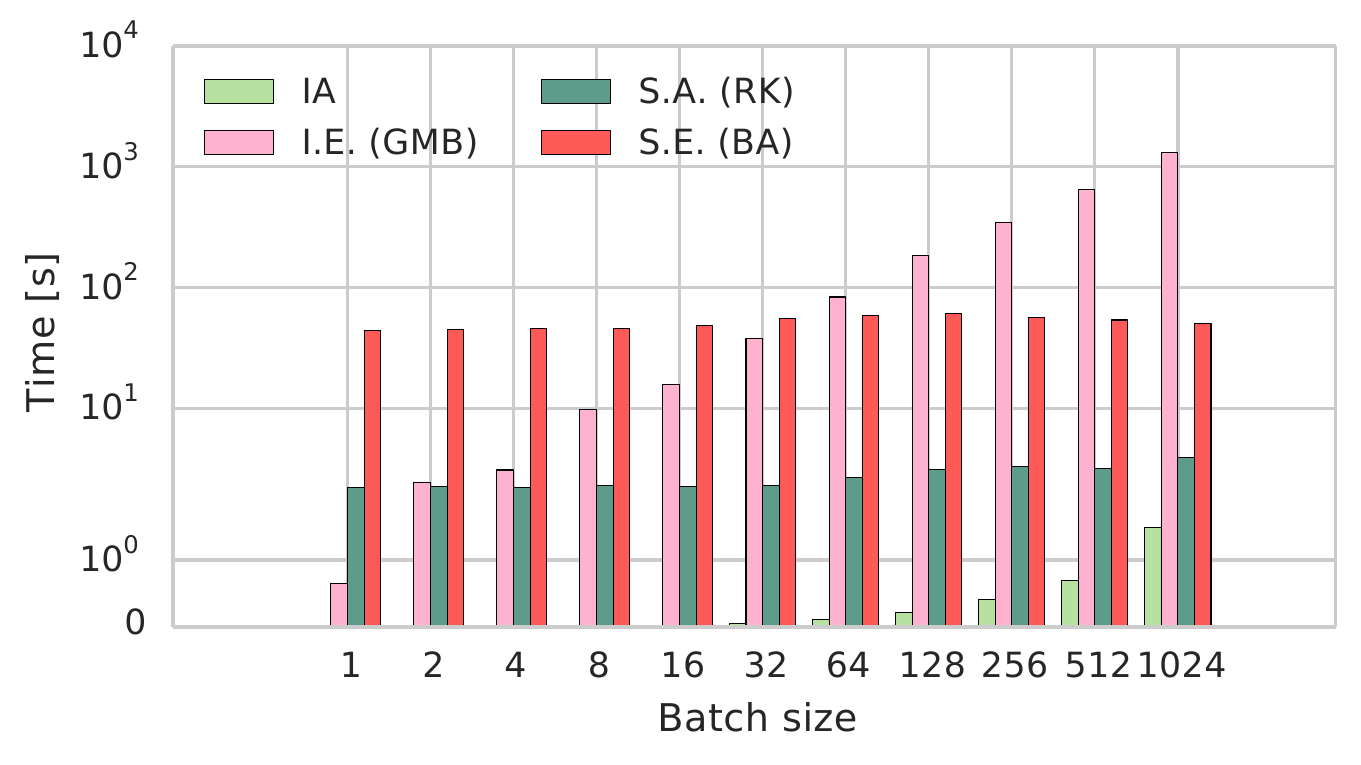}
\includegraphics[width = 0.45\textwidth]{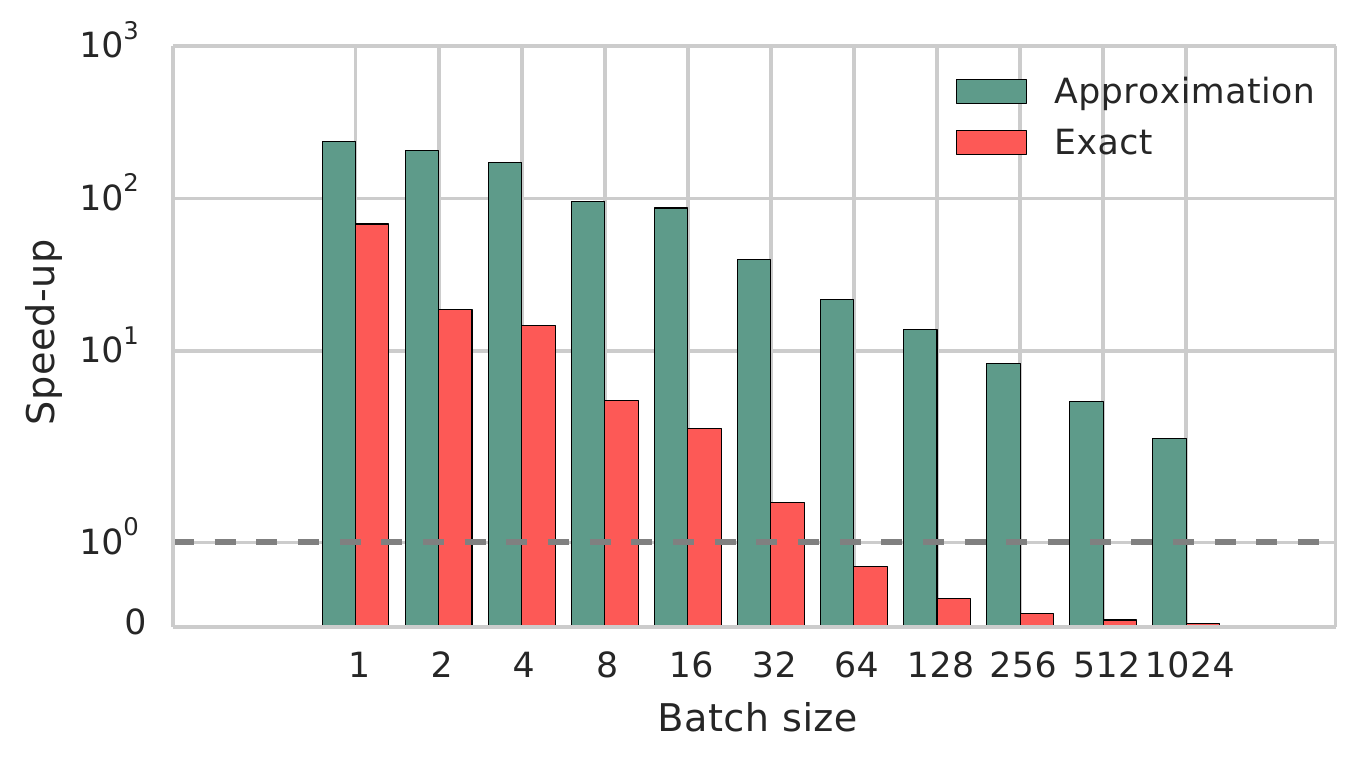}
\caption{Running times and speedups on \texttt{PGPgiantcompo}, with $\epsilon = 0.05$ and with batches of different sizes. Left: Running times of static exact (\textsf{BA}), static approximation (\textsf{RK}), incremental exact (\textsf{GMB}) and our incremental approximation \textsf{IA}. Right: Speedups of \textsf{GMB} on \textsf{BA} (exact) and of \textsf{IA} on \textsf{RK} (approximation).} 
\label{pgp}
\end{center}
\vspace{-4ex}
\end{figure}
\begin{figure}[h!]
\begin{center}
\includegraphics[width = 0.45\textwidth]{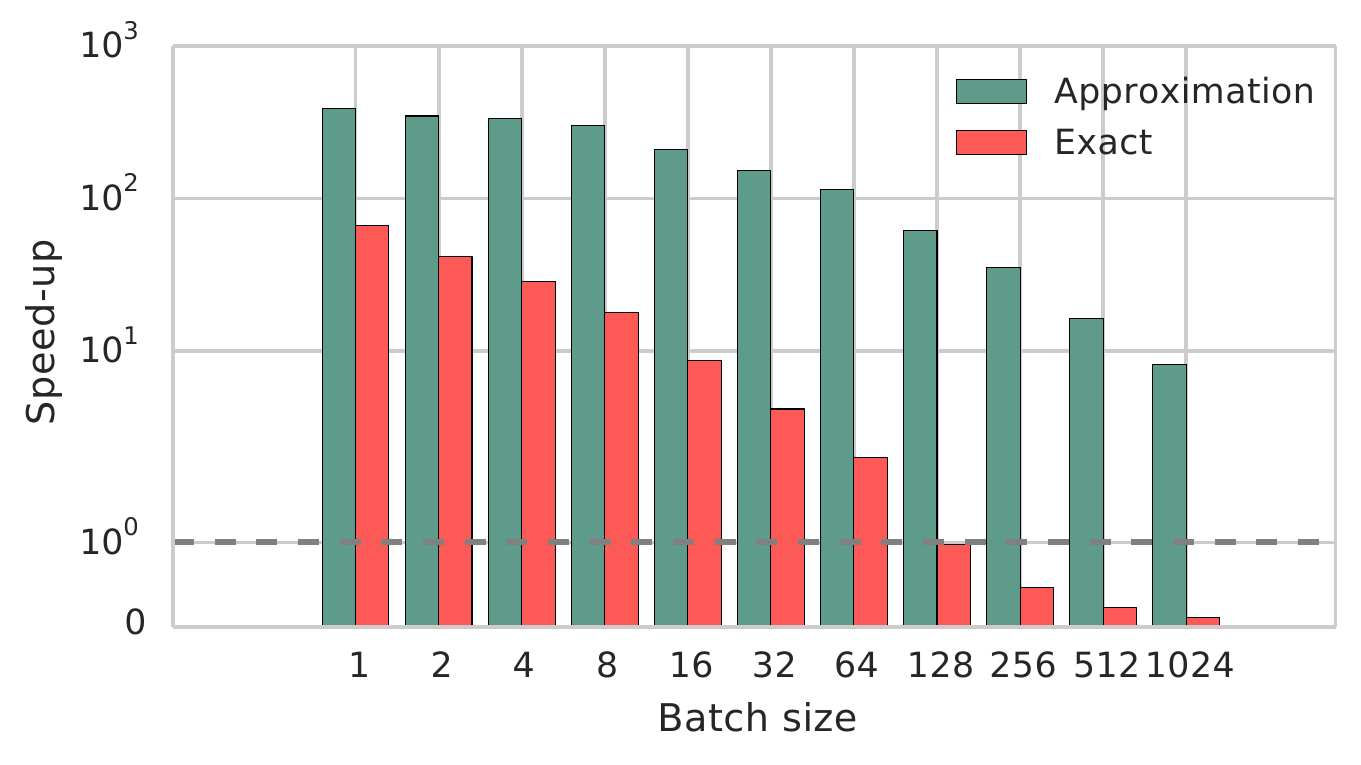}
\includegraphics[width = 0.45\textwidth]{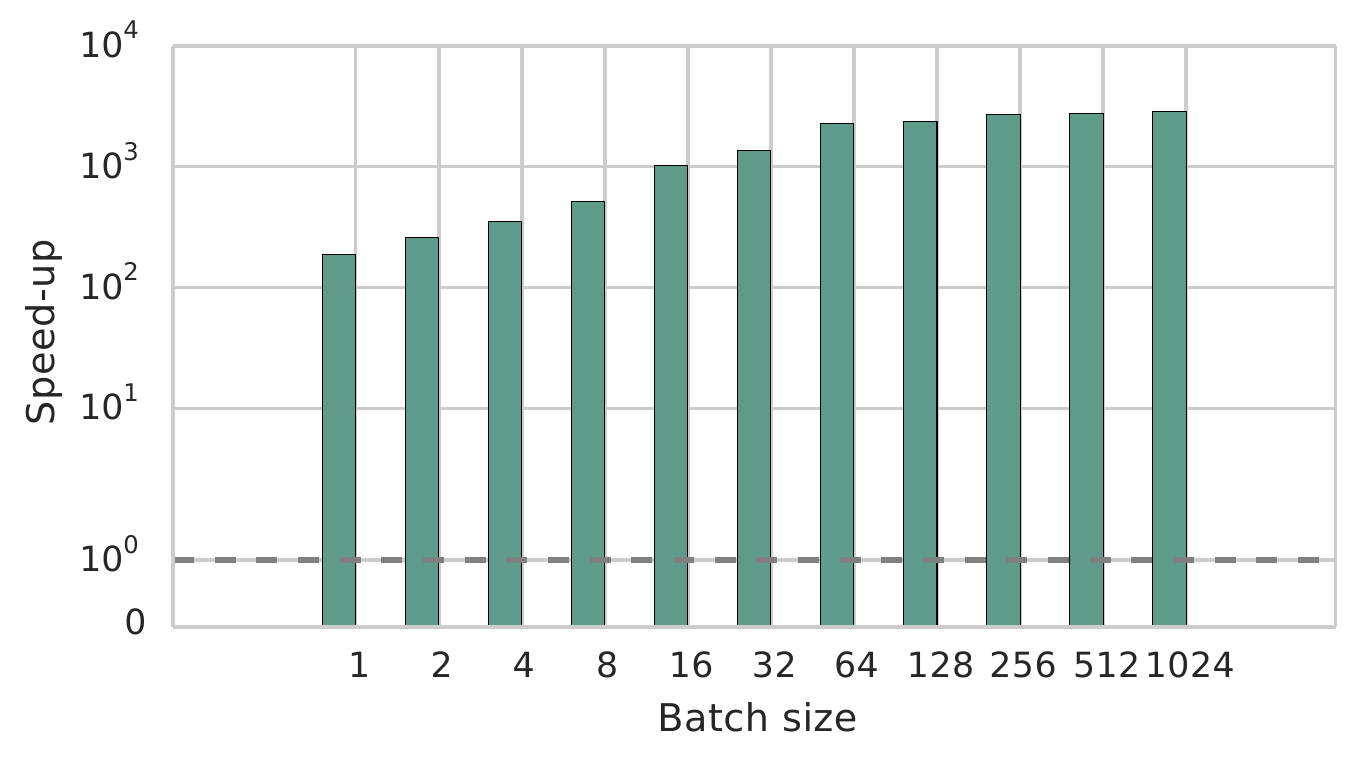}
\caption{Speedups on Dorogovtsev-Mendes synthetic graphs ($m= 40k$), with $\epsilon = 0.05$ with batches of different sizes. Left: comparison of the speedups of \textsf{GMB} on \textsf{BA} (exact) and of \textsf{IA} on \textsf{RK} (approximation). Right: speedups of \textsf{IA} on \textsf{GMB}.} 
\label{small_synth}
\end{center}
\vspace{-4ex}
\end{figure}

Since the exact algorithms are not scalable, for the comparison on larger networks (Table~\ref{table:graphs}) we used only \textsf{RK} and our algorithms. 
Figure~\ref{fig:speedups_real} (left) reports the speedups of \da on \rk in real graphs using real dynamics. Although some fluctuations can be noticed, the speedups tend to decrease as the batch size increases. We can attribute fluctuations to two main factors: First, different batches can affect areas of $G$ of varying sizes, influencing also the time required to update the SSSPs. Second, changes in the \vd approximation can require to sample new paths and therefore increase the running time of \da (and \daw). Nevertheless, \da is significantly faster than recomputation on all networks and for every tested batch size. 
Analogous results are reported in Figure~\ref{fig:random} (left) for random dynamics. 
Table~\ref{table:speedups_real} summarizes the running times of \da and its speedups on \rk with batches of size 1 and 1024 in unweighted graphs, under both real and random dynamics. Even on the larger graphs (\texttt{arXivCitations} and \texttt{englishWikipedia}) and on large batches, \da requires at most a few seconds to recompute the BC scores, whereas \rk requires about one hour for \texttt{englishWikipedia}. 
The results on weighted graphs are shown in Table~\ref{table:weighted}. 
In both real dynamics and random updates, the speedups vary between $\approx 50$ and $\approx 6 \cdot 10^3$ for single-edge updates and between $\approx 5$ and $\approx 75$ for batches of size 1024.

On hyperbolic graphs (Figure~\ref{fig:speedups_real}, right), the speedups of \da on \rk increase with the size of the graph. Table~\ref{table:hyperbolic} contains the precise running times and speedups on batches of 1 and 1024 edges. The speedups vary between $\approx 100$ and $\approx 3 \cdot 10^5$ for single-edge updates and between $\approx 3$ and $\approx 5 \cdot 10^3$ for batches of 1024 edges. 

Some graphs of Table~\ref{table:graphs} can also be interpreted as directed (\ie,  \texttt{repliesDigg}, \texttt{emailSlashdot}, \texttt{emailLinux}, \texttt{facebookPosts}, \texttt{emailEnron}, \texttt{arXivCitations}). We therefore test \dad on the directed version of the networks, using real dynamics. Also on directed networks, using \dad is
always faster than recomputation with \rk, by orders of magnitude on small batches. Figure~\ref{fig:random} (right) shows the speedups of \dad on \rk on the \texttt{facebookPosts} graph.

\begin{table*}
\begin{center}
\begin{scriptsize}
  \begin{tabular}{  l | r | r | r | r | r | r | r | r |}
\cline{2-9}  
 & \multicolumn{4}{  c |  }{Real}& \multicolumn{4}{  c |  }{Random} \\ \cline{2-9}
 &\multicolumn{2}{  c |  }{Time [s]}& \multicolumn{2}{  c |  }{Speedups} &\multicolumn{2}{  c |  }{Time [s]}& \multicolumn{2}{  c |  }{Speedups} \\ \cline{1-9}
\multicolumn{1}{| l|}{Graph} & $|\beta|=1$ & $|\beta|=1024$ & $|\beta|=1$ & $|\beta|=1024$  & $|\beta|=1$ & $|\beta|=1024$ & $|\beta|=1$ & $|\beta|=1024$ \\\cline{1-9}
\multicolumn{1}{| l|}{\texttt{repliesDigg}} 		& 0.078	& 1.028	& 76.11	& 5.42	& 0.008	& 0.832	& 94.00	& 4.76	\\ \cline{1-9}		
\multicolumn{1}{| l|}{\texttt{emailSlashdot}} 	& 0.043	& 1.055	& 219.02	& 9.91	& 0.038	& 1.151	& 263.89	& 28.81	\\ \cline{1-9}		
\multicolumn{1}{| l|}{\texttt{emailLinux}} 		& 0.049	& 1.412	& 108.28	& 3.59	& 0.051	& 2.144	& 72.73	& 1.33	\\ \cline{1-9}				
\multicolumn{1}{| l|}{\texttt{facebookPosts}} 	& 0.023	& 1.416	& 527.04	& 9.86	& 0.015	& 1.520	& 745.86	& 8.21	\\ \cline{1-9}	
\multicolumn{1}{| l|}{\texttt{emailEnron}} 		& 0.368	& 1.279	& 83.59	& 13.66	& 0.203	& 1.640	& 99.45	& 9.39	\\ \cline{1-9}				
\multicolumn{1}{| l|}{\texttt{facebookFriends}} 	& 0.447	& 1.946	& 94.23	& 18.70	& 0.448	& 2.184	& 95.91	& 18.24	\\ \cline{1-9}			
\multicolumn{1}{| l|}{\texttt{arXivCitations}} 	& 0.038	& 0.186	& 2287.84	& 400.45	& 0.025	& 1.520	& 2188.70	& 28.81	\\ \cline{1-9}			
\multicolumn{1}{| l|}{\texttt{englishWikipedia}} 	& 1.078	& 6.735	& 3226.11	& 617.47	& 0.877	& 5.937	& 2833.57	& 703.18	\\ \cline{1-9}		
  \end{tabular}
  \end{scriptsize}
\end{center}
  \caption{Times and speedups of \da on \rk in unweighted real graphs under real dynamics and random updates, for batch sizes of 1 and 1024.}
  \label{table:speedups_real}
%   \vspace{-6ex}
\end{table*} 
\begin{figure}[h!]
\begin{center}
\includegraphics[width = 0.49\textwidth]{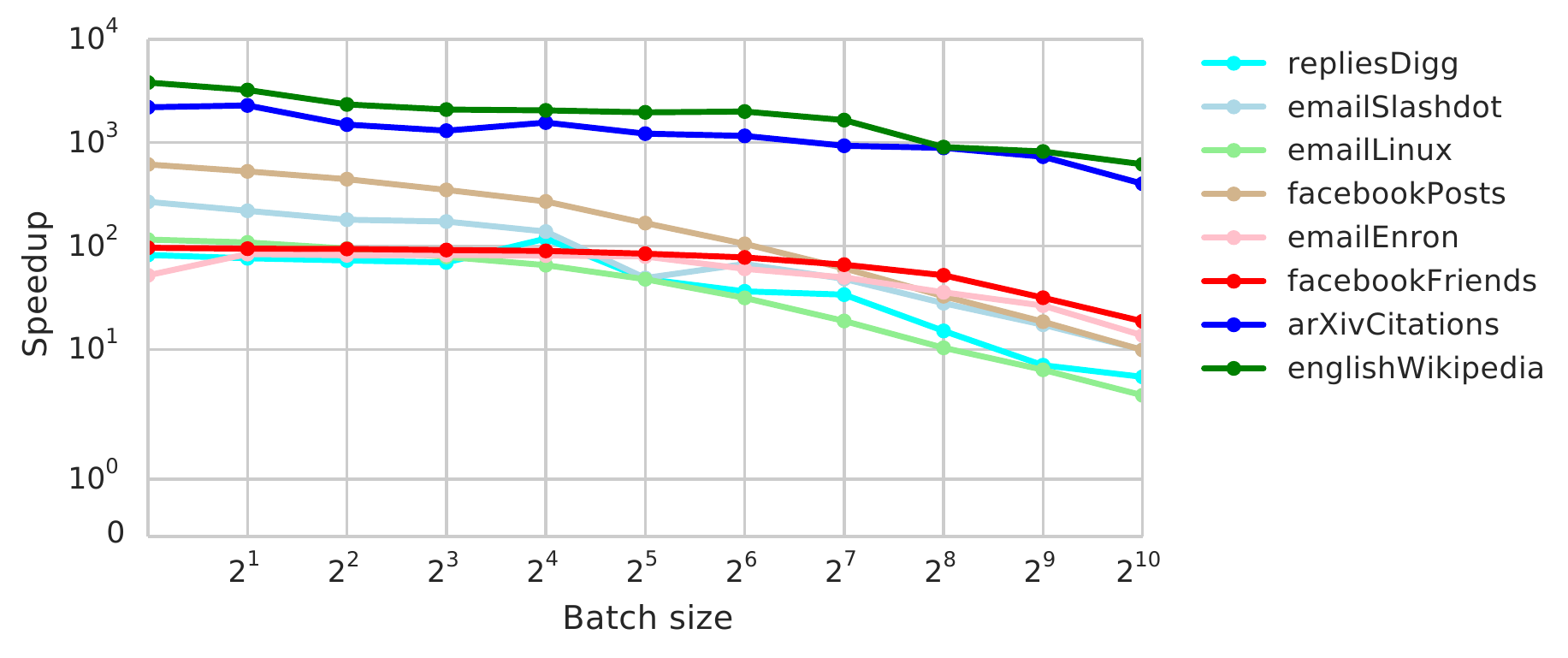}
\includegraphics[width = 0.49\textwidth]{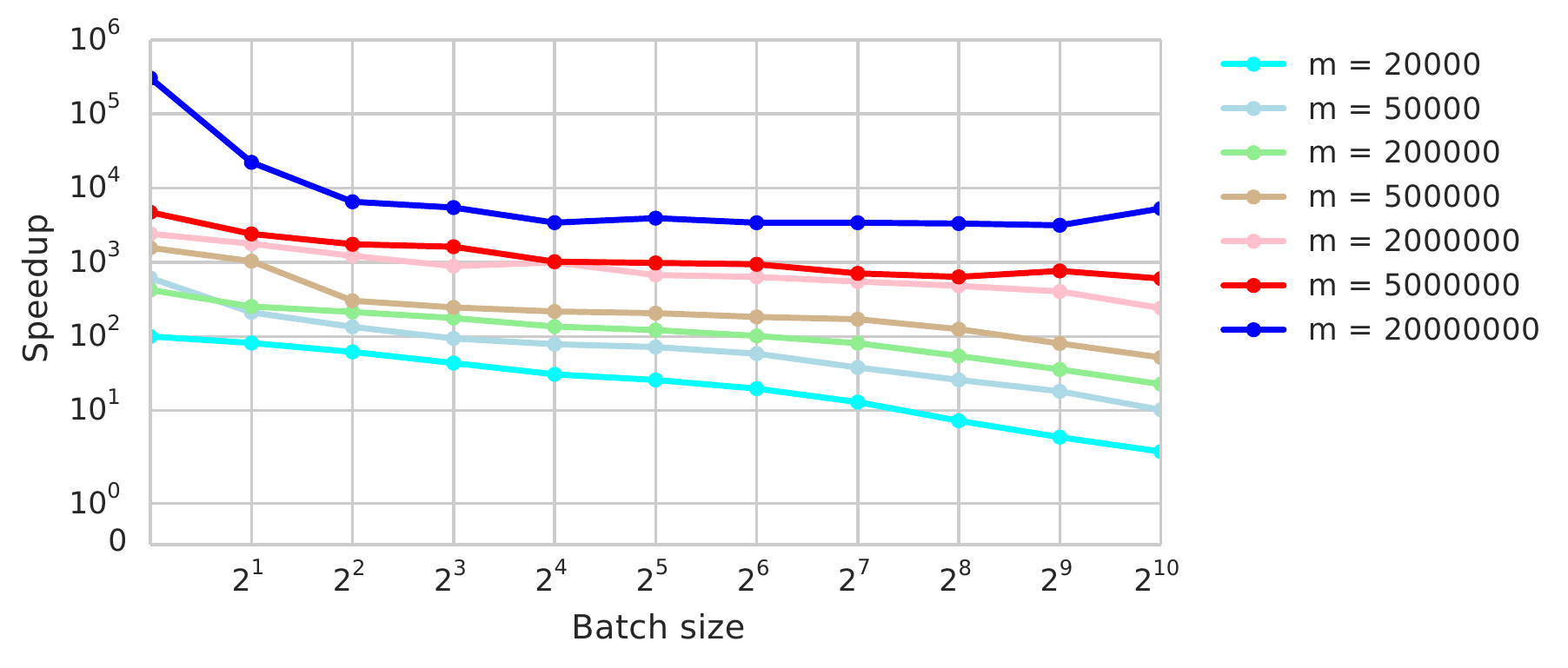}
\caption{Speedups of \da on \rk, with $\epsilon = 0.05$ and with batches of different sizes. Left: real unweighted networks using real dynamics. Right: hyperbolic unit-disk graphs of different sizes.} 
\label{fig:speedups_real}
\end{center}
\vspace{-4ex}
\end{figure}
% \vspace{-4ex}
\begin{figure}[h!]
\begin{center}
\includegraphics[width = 0.55\textwidth]{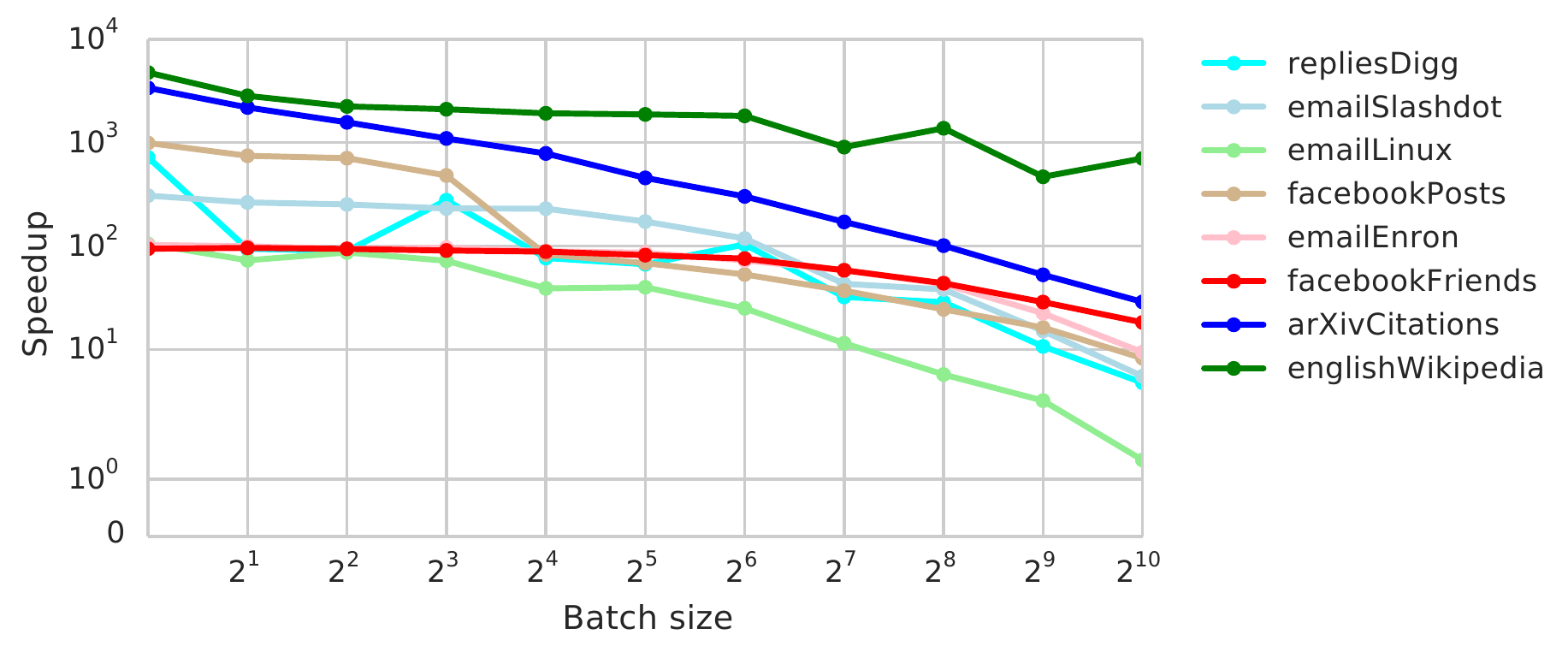}
\includegraphics[width = 0.41\textwidth]{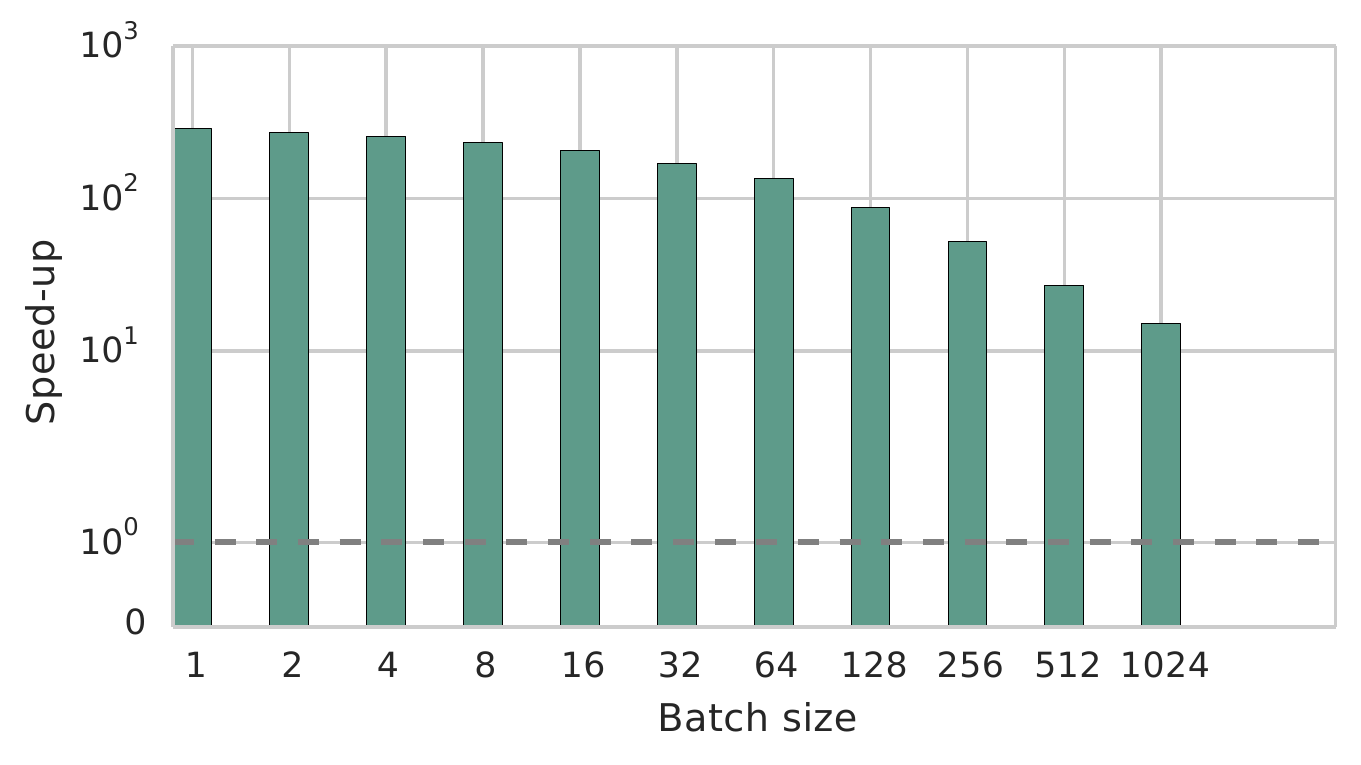}
\caption{Left: speedups of \da on \rk in real unweighted graphs under random updates. Right: Speedups of \dad on \rk on the \texttt{facebookPosts} directed graph using real dynamics.} 
\label{fig:random}
\end{center}
%\vspace{-4ex}
\end{figure}

% \vspace{-2ex}
 \begin{table*}[h]
 % \vspace{-4ex}
\begin{center}
\begin{scriptsize}
  \begin{tabular}{  l | r | r | r | r | r | r | r | r |}
\cline{2-9}  
 & \multicolumn{4}{  c |  }{Real}& \multicolumn{4}{  c |  }{Random} \\ \cline{2-9}
 &\multicolumn{2}{  c |  }{Time [s]}& \multicolumn{2}{  c |  }{Speedups} &\multicolumn{2}{  c |  }{Time [s]}& \multicolumn{2}{  c |  }{Speedups} \\ \cline{1-9}
\multicolumn{1}{| l|}{Graph} & $|\beta|=1$ & $|\beta|=1024$ & $|\beta|=1$ & $|\beta|=1024$  & $|\beta|=1$ & $|\beta|=1024$ & $|\beta|=1$ & $|\beta|=1024$ \\\cline{1-9}
\multicolumn{1}{| l|}{\texttt{repliesDigg}} 		& 0.053	& 3.032	& 605.18	& 14.24	& 0.049	& 3.046	& 658.19	& 14.17	\\ \cline{1-9}		
\multicolumn{1}{| l|}{\texttt{emailSlashdot}} 	& 0.790	& 5.387	& 50.81	& 16.12	& 0.716	& 5.866	& 56.00	& 14.81	\\ \cline{1-9}		
\multicolumn{1}{| l|}{\texttt{emailLinux}} 		& 0.324	& 24.816	& 5780.49	& 75.40	& 0.344	& 24.857	& 5454.10	& 75.28	\\ \cline{1-9}				
\multicolumn{1}{| l|}{\texttt{facebookPosts}} 	& 0.029	& 6.672	&  2863.83& 11.42	& 0.029	& 6.534	& 2910.33	& 11.66	\\ \cline{1-9}	
\multicolumn{1}{| l|}{\texttt{emailEnron}} 		& 0.050	& 9.926	& 3486.99	& 24.91	& 0.046	& 50.425	& 3762.09	& 4.90	\\ \cline{1-9}					
  \end{tabular}
  \end{scriptsize}
\end{center}
  \caption{Times and speedups of \daw on \rk in weighted real graphs under real dynamics and random updates, for batch sizes of 1 and 1024.}
  \label{table:weighted}
%  \vspace{-4ex}
\end{table*} 
  \vspace{-2ex}
 \begin{table*}[h]
%   \vspace{-4ex}
\begin{center}
\begin{scriptsize}
  \begin{tabular}{  l | r | r | r | r |}
\cline{2-5}  
 & \multicolumn{4}{  c |  }{Hyperbolic} \\ \cline{2-5}
 &\multicolumn{2}{  c |  }{Time [s]}& \multicolumn{2}{  c |  }{Speedups} \\ \cline{1-5}
\multicolumn{1}{| l|}{Number of edges} & $|\beta|=1$ & $|\beta|=1024$ & $|\beta|=1$ & $|\beta|=1024$ \\\cline{1-5}
\multicolumn{1}{| l|}{$m = 20000$} 				& 0.005	& 0.195	& 99.83	& 2.79	\\ \cline{1-5}		
\multicolumn{1}{| l|}{$m = 50000$} 				& 0.002	& 0.152	& 611.17	& 10.21	\\ \cline{1-5}
\multicolumn{1}{| l|}{$m = 200000$} 				& 0.015	& 0.288	& 422.81	& 22.64	\\ \cline{1-5}		
\multicolumn{1}{| l|}{$m = 500000$} 				& 0.012	& 0.339	& 1565.12	& 51.97	\\ \cline{1-5}		
\multicolumn{1}{| l|}{$m = 2000000$} 			& 0.049	& 0.498	& 2419.81	& 241.17	\\ \cline{1-5}		
\multicolumn{1}{| l|}{$m = 5000000$} 			& 0.083	& 0.660	& 4716.84	& 601.85	\\ \cline{1-5}			
\multicolumn{1}{| l|}{$m = 20000000$} 			& 0.006	& 0.401	& 304338.86	& 5296.78	\\ \cline{1-5}							
  \end{tabular}
  \end{scriptsize}
\end{center}
  \caption{Times and speedups of \da on \rk in hyperbolic unit-disk graphs, for batch sizes of 1 and 1024.}
  \label{table:hyperbolic}
\end{table*}

To summarize, our results show that our dynamic algorithms are faster than recomputation with \rk in all the tested instances, even when large batches of 1024 edges are applied to the graph. With small batches, the algorithms are always orders of magnitude faster than \rk, often with running times of fraction of seconds or seconds compared to minutes or hours. Such high speedups are made possible by the efficient update of the sampled shortest paths, which limits the recomputation to the nodes that are actually affected by the batch. Also, processing the edges in batches, we avoid to update multiple times nodes that are affected by several edges of the batch.

\section{Conclusions}
Because betweenness centrality considers all shortest paths between pairs of nodes, its exact computation is out of reach for the large complex networks that come up in many applications today. 
However, approximate scores obtained by sampling paths are often sufficient to identify the most important nodes and rank nodes in an order that is very similar to exact BC. 
Since many applications are concerned with rapidly evolving networks, we have explored whether a dynamic approach -- which updates paths when a batch of edges arrive or are deleted from the network -- is more efficient than recomputing BC scores from scratch.
We introduce the first dynamic betweenness approximation algorithms with provable error guarantee. They work for weighted, unweighted, directed and undirected graphs.
Our theoretical results are confirmed by experiments on a diverse set of real-world networks with both real and simulated dynamics,
which show the effectiveness of our approach.

%Their common method is based on the static \textsf{RK} approximation algorithm and inherits its provable bound on the maximum error.
%Through reimplementation and experiments, we show that previously proposed dynamic BC algorithms scale badly, also due to an $\Omega(n^2)$ memory footprint.
%Combining a dynamic approach with approximation, our method is more scalable and is the first dynamic one applicable to networks with millions of edges (without using external memory).
The dynamic updating of paths implies a higher memory footprint, but also enables significant speedups compared to recomputation (\eg factor 100 for a batch of 1024 edge insertions).
The scalability of our algorithms is primarily limited by the available memory. Each sampled path requires $O(n)$ memory and the number of required samples grows quadratically as the error bound $\epsilon$ is tightened. 
This leaves the user with a tradeoff between running time and memory consumption on the one hand and BC score error on the other hand.
However, our experiments indicate that even a relatively high error bound (\eg $\epsilon = 0.1$) for the BC scores preserves the ranking for the more important nodes reasonably well.

We studied sequential implementations for simplicity and comparability with related work, but parallelization is possible, part
of future work, and can yield further speedups in practice.
Our implementations are based on \textit{NetworKit}\footnotemark, the open-source framework for large-scale network analysis.
Most of the source code used for this paper is already available in \textit{NetworKit}, the remaining code will follow in upcoming
releases of the package.

An interesting open problem remains the update of BC approximations in scenarios where also nodes can be inserted and deleted. The presence of a new node (or the deletion of an existing one) would modify the probability distribution that regulates the path sampling, requiring the necessity of new techniques.

\footnotetext{\url{http://networkit.iti.kit.edu}}

\bigskip

\renewcommand{\baselinestretch}{0.25}
\begin{scriptsize}
\textbf{Acknowledgements.}
This work is partially supported by DFG grant ME-3619/3-1 (FINCA) within the SPP 1736 \emph{Algorithms for Big Data} and by DFG grant ME-3619/2-1 (TEAM).
We thank Moritz von Looz for providing the synthetic dynamic networks and the numerous contributors to 
the \textit{NetworKit} project. We also thank Matteo Riondato for his constructive comments
on earlier versions of the material presented in this paper.
\end{scriptsize}
\renewcommand{\baselinestretch}{1.0}%%%%% bib %%%%%%%%%
\bibliographystyle{abbrv}
\bibliography{references}
%%%%%%%%%%%%%%%%

%%%%% appendix %%%%%%
\pagebreak

\end{document}